\documentclass[runningheads]{llncs}
\usepackage[T1]{fontenc}
\usepackage{graphicx}
\usepackage{hyperref}
\usepackage{color}

\usepackage{macros}

\begin{document}
\title{Dual Adjunction Between \texorpdfstring{$\Omega$}{Omega}-Automata and Wilke Algebra Quotients}
%
%
\author{Anton Chernev\inst{1}\orcidID{0009-0002-6232-5604} \and
Helle~Hvid Hansen\inst{1}\orcidID{0000-0001-7061-1219} \and
Clemens Kupke\inst{2}\orcidID{0000-0002-0502-391X}}
\authorrunning{A. Chernev et al.}
%
\institute{University of Groningen, Netherlands \\ \email{a.chernev@rug.nl} \\ \email{h.h.hansen@rug.nl} \and
University of Strathclyde, United Kingdom \\
\email{clemens.kupke@strath.ac.uk}}
\maketitle              
\begin{abstract}
$\Omega$-automata and Wilke algebras are formalisms for characterising $\omega$-regular languages via their ultimately periodic words. $\Omega$-automata read finite representations of ultimately periodic words, called lassos, and they are a subclass of lasso automata.
We introduce lasso semigroups as a generalisation of Wilke algebras that mirrors how lasso automata generalise $\Omega$-automata, and we show that finite lasso semigroups characterise regular lasso languages. 
We then show a dual adjunction between lasso automata and quotients of the free lasso semigroup with a recognising set, and as our main result we show that this dual adjunction restricts to one between $\Omega$-automata and quotients of the free Wilke algebra with a recognising set. 

\keywords{Infinite words \and $\omega$-regular languages \and Ultimately periodic words \and $\Omega$-automata \and Wilke algebra \and Coalgebra}
\end{abstract}
\section{Introduction}
\emph{$\Omega$-automata} \cite{CianciaVenema2012StreamAutomataAreCoalgebras,CianciaVenema2019OmegaAutomataACoalgebraicPerspective} were introduced as a way of capturing $\omega$-regular languages coalgebraically \cite{Rutten:TCS2000}. This is based on two main observations. First, every $\omega$-regular language $L$ is determined by its set of \emph{ultimately periodic words} $ \{uv^\omega \mid uv^\omega \in L \}$ (e.g., \cite[Fact~1]{CalbrixNivatPodelski1994UltimatelyPeriodicWords}). Second, for every $\omega$-regular language $L$, the language $\{ u\$v \mid uv^\omega \in L \}$ is regular \cite[Prop.~4]{CalbrixNivatPodelski1994UltimatelyPeriodicWords}. $\Omega$-automata run on \emph{lassos}, which are pairs of finite words $(u, v)$ representing $uv^\omega$. Thus every $\omega$-regular language $L$ is identified by an $\Omega$-automaton accepting the \emph{lasso language} $\{ (u,v) \mid uv^\omega \in L \}$.
The fact that $\Omega$-automaton bisimilarity corresponds to lasso language equivalence \cite{CianciaVenema2012StreamAutomataAreCoalgebras} enables algorithms for deciding language equivalence of $\Omega$-automata, as well as minimisation algorithms using partition refinement \cite{CianciaVenema2019OmegaAutomataACoalgebraicPerspective} or Brzozowski-style via dual adjunctions \cite[Ch.~8]{Cruchten2022TopicsInOmegaAutomata}.

$\Omega$-automata are defined as the subclass of \emph{lasso automata} \cite{CianciaVenema2012StreamAutomataAreCoalgebras}
that satisfy two conditions (circularity and coherence) which ensure that $\Omega$-automata accept
lasso languages that are \emph{saturated}, meaning that $u_1v_1^\omega = u_2v_2^\omega$ implies $(u_1,v_1)$ and $(u_2, v_2)$ are both accepted or both rejected. This is required in order for $\Omega$-automata languages to correspond to $\omega$-regular languages.
Lasso automata (accepting non-saturated languages) are of independent interest. They are studied in \cite{AngluinFisman2016LearningRegularOmegaLang} (under the name FDFAs) in the context of learning $\omega$-regular languages. There it is shown that certain lasso automaton representations of $\omega$-regular languages can be factorially smaller than their $\Omega$-automaton representations\footnote{In the terminology of \cite{AngluinFisman2016LearningRegularOmegaLang}, \emph{syntactic} and \emph{recurrent FDFAs} can be smaller than $L_\$$.}.

Our motivation for the present work is to better understand the mathematical connections between the coalgebraic theory of $\omega$-regular languages, given by $\Omega$-automata, and the algebraic theory of $\omega$-regular languages, given by \emph{algebraic recognition} via \emph{Wilke algebras} \cite[Sec.~2.5]{PerrinPin2004InfiniteWords}.
In the setting of finite words, \cite{Planting2013AutomataToMonoids,Cruchten2024OnTransitionConstructionsArxiv} show an adjunction between deterministic finite automata, on the coalgebra side, and monoid congruences \cite{Pin:MathematicalFoundationsOfAutomataTheory}, on the algebra side.
We are interested in establishing a similar result for $\Omega$-automata and Wilke algebras. 
In \cite[Ch.~5]{Cruchten2022TopicsInOmegaAutomata}, a construction is given from $\Omega$-automata to Wilke algebra homomorphisms that recognise the same language. However, the construction is only defined on objects, and the converse direction is not treated.

In this paper, we exhibit a dual adjunction between $\Omega$-automata and \emph{extended Wilke algebras}. We define the latter as surjective homomorphisms with the freely generated Wilke algebra as their domain, together with a recognising set. 
We obtain this adjunction as the restriction of another adjunction, between lasso automata and a new type of algebraic structures that we call \emph{extended lasso semigroups}. We define lasso semigroups by omitting the \emph{circularity} and \emph{coherence} axioms of Wilke algebras. The lasso automaton adjunction looks as follows:
\begin{equation}
    \begin{tikzcd}
        {\text{Ext Lasso Sgp}} & \bot & {\text{Lasso Aut}} & \bot & {\text{Lasso Aut}^{\mathrm{op}}}
    	\arrow["\Aut", curve={height=-15pt}, from=1-1, to=1-3]
    	\arrow["\Alg", curve={height=-15pt}, from=1-3, to=1-1]
    	\arrow["\Rev", curve={height=-15pt}, from=1-3, to=1-5]
    	\arrow["{\op\Rev}", curve={height=-15pt}, from=1-5, to=1-3]
    \end{tikzcd}
    \label{diag:lassoAutAdjInformal}
\end{equation}
On the right, $\Rev \dashv \op \Rev$ is the transition-reversal adjunction described in~\cite[Sec.~8.1]{Cruchten2022TopicsInOmegaAutomata}. On the left, $\Aut$ and $\Alg$ are new constructions between extended lasso semigroups and lasso automata that \emph{reverse} the accepted language. In particular, $\Alg$ is different from the construction in \cite[Ch.~5]{Cruchten2022TopicsInOmegaAutomata}. 
By taking suitable restrictions of the functors in Diagram \labelcref{diag:lassoAutAdjInformal}, we obtain the adjunction:
\begin{equation}
    \begin{tikzcd}
        {\text{Ext Wilke Alg}} & \bot & {\text{$\Omega^\rev$-Aut}} & \bot & {\text{$\Omega$-Aut}^{\mathrm{op}}}
        \arrow["\Aut", curve={height=-15pt}, from=1-1, to=1-3]
        \arrow["\Alg", curve={height=-15pt}, from=1-3, to=1-1]
        \arrow["\Rev", curve={height=-15pt}, from=1-3, to=1-5]
        \arrow["{\op\Rev}", curve={height=-15pt}, from=1-5, to=1-3]
    \end{tikzcd}
    \label{diag:OmegaAutAdjInformal}
\end{equation}
Here $\Omega^\rev$-automata (in words, \emph{reverse-$\Omega$-automata}) are a new type of lasso automata that correspond to the reverse of $\Omega$-automata.

Furthermore, we show that lasso semigroups provide an algebraic characterisation of lasso languages beyond saturated languages. That is, homomorphisms into finite lasso semigroups recognise precisely the regular lasso languages.

We note that
dual adjunctions between coalgebras and algebras have been shown in \cite{Bezhanishvili2023MinimizationInLogicalForm,BonchiEtAl2014AlgebraCoalgebraDualityInBrzozowski,Rot16:CoalgMinInitialityFinality} to give rise to abstract minimisation algorithms for a wide range of automata that operate on finite words. Similar results have been shown for $\Omega$-automata in \cite[Ch.~8]{Cruchten2022TopicsInOmegaAutomata}. These dual adjunctions are of a different nature from the ones studied here, but they have also served as motivation and inspiration. 

The paper is organised as follows. In \Cref{sec:preliminaries} we collect basic definitions and notation on lasso automata, $\Omega$-automata and Wilke algebras. In \Cref{sec:algRecognitionLassoLanguages} we introduce lasso semigroups, define the maps $\Aut, \Alg$ and $\Rev$ and use them to show that finite lasso semigroups recognise $\omega$-regular languages (\Cref{thm:lassoAutomataAdjunctionChain}). In \Cref{sec:adjunctionLassoAutLassoAlg} we extend these maps to functors and prove the adjunction from Diagram \labelcref{diag:lassoAutAdjInformal} (\Cref{thm:lassoAutomataAdjunctionChain}). 
We use it to derive the adjunction from Diagram \labelcref{diag:OmegaAutAdjInformal} (\Cref{thm:OmegaAutAdjunction}) in \Cref{sec:adjunctionOmegaAutomata}. 
At the end of \Cref{sec:adjunctionLassoAutLassoAlg,sec:adjunctionOmegaAutomata}, we briefly discuss how our functors relate minimal automata and maximal Wilke algebra quotients, and applications of the adjunction.
We conclude with a summary and a discussion of related and future work in \Cref{sec:conclusion}.


\section{Preliminaries}
\label{sec:preliminaries}

We assume familiarity with basic concepts from category theory, such as categories, functors and adjunctions (see, e.g., \cite{Awodey2006CategoryTheory,MacLane1971CategoriesWorking}), and from the theory of $\omega$-regular languages (e.g., \cite{GreenBook}).

\subsection{Languages of Infinite Words}
Throughout this paper, we fix a set of symbols $\Sigma = \{ a, b, \dotsc \}$, called an \emph{alphabet}. Let $\words$ denote the set of \emph{finite words} over $\Sigma$ and $\newords$ denote the set of \emph{non-empty words}. We have $\newords = \words \setminus \{\epsilon\}$, where $\epsilon$ stands for the empty word. We often use the notation $au$ or $ua$, where $a \in \Sigma$ and $u \in \words$, for an arbitrary non-empty word. An \emph{infinite word} over $\Sigma$ is a sequence of elements of $\Sigma$ of length $\omega$. An \emph{ultimately periodic word} is an infinite word of the form $uv^\omega \coloneqq uvv\dotsc$, and the set of all ultimately periodic words is written as $\upwords$. A \emph{lasso} is a pair $(u, v) \in \words\times \newords$, with the set of all lassos written as $\lassos$. Intuitively, the lasso $(u,v)$ represents the ultimately periodic word $uv^\omega$. A \emph{lasso language} is a subset of $\lassos$. Similarly, a \emph{language of infinitely periodic words} is a subset of $\upwords$. A lasso language $L$ is \emph{saturated} if $u_1v_1^\omega = u_2v_2^\omega$ implies $(u_1,v_1) \in L \iff (u_2,v_2)\in L$.

Given some $u = a_1\dotsc a_n \in \words$, we write $u^\rev \coloneqq a_n\dotsc a_1$ for the \emph{reverse word of $u$}. While infinite words cannot be reversed, lassos can, because they are finite objects. Thus define the \emph{reverse of a lasso} $(u,av)$ as the lasso $(u,av)^\rev \coloneqq (v^\rev, au^\rev)$. On the level of languages, given a lasso language $L$, we write $L^\rev \coloneqq \{ (u,av)^\rev \mid (u,av) \in \lassos \}$ for the \emph{reverse lasso language of $L$}.

\subsection{Lasso Automata and $\Omega$-Automata}
\emph{Lasso automata} were introduced in \cite{CianciaVenema2012StreamAutomataAreCoalgebras,CianciaVenema2019OmegaAutomataACoalgebraicPerspective} as acceptors of lasso languages.
\begin{definition}[Lasso automaton {\cite{CianciaVenema2012StreamAutomataAreCoalgebras}}]
    A \emph{lasso automaton} is a tuple $A = (X, Y, q,\rho, \sigma, \xi, F)$ where:
    \begin{itemize}
        \item $X$ and $Y$ are disjoint finite sets whose elements are called \emph{states};
        \item $q$ is a state in $X$ called the \emph{initial state};
        \item the functions $\rho: X \times \Sigma \to X$, $\sigma: X \times \Sigma \to Y$ and $\xi: Y \times \Sigma \to Y$ are called \emph{transition functions};
        \item $F$ is a subset of $Y$ whose elements are called \emph{final states}.
    \end{itemize}
\end{definition}
The transition function $\rho$ will often be tacitly used as a function from $X \times \Sigma^*$ to $X$ in the standard way. That is, $\rho(x, \epsilon) \coloneqq x$ and $\rho(x, ua) \coloneqq \rho(\rho(x, u), a)$. This applies analogously to $\xi$.

The lasso automaton structure allows for a natural definition of lasso acceptance. A lasso $(u,av)$ is read as follows: $\rho$ transitions read $u$, $\sigma$ reads $a$, and $\xi$ reads $v$. Formally, given a lasso automaton $A = (X, Y, q,\rho, \sigma, \xi, F)$, define \emph{the lasso language accepted by $A$} as $\Lasso(A) \coloneqq \{ (u,av)\in\lassos \mid \xi(\sigma(\rho(q,u),a),v) \in F \}$. A lasso language is called \emph{regular} if it is accepted by some \emph{finite} lasso automaton.

\begin{example}
    In \Cref{fig:lassoAut} we see two examples of lasso automata for $\Sigma = \{a,b\}$. 
    It can easily be verified that $\Lasso(A_1) = \{ (u, bv) \mid u,v \in \words \}$
    and $\Lasso(A_2) = \{ (ub, a^n) \mid u \in \words, n \in \omega \}$. Note that $\Lasso(A_1)$ is not saturated, since it contains $(\epsilon,ba)$, but not $(b,ab)$.
\end{example}
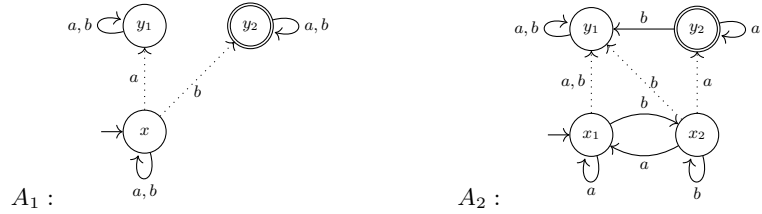
\begin{figure}
    \centering
    $A_1:$ \begin{tikzpicture}[scale=0.7, every node/.style={scale=0.7}]
        \node[state, initial] at (0,0) (x1) {$x$};
        \node[state] at (0, 2) (y1) {$y_1$};
        \node[state, accepting] at (2, 2) (y2) {$y_2$};
        \draw (x1) edge[loop below] node{$a,b$} (x1)
            (x1) edge[left, dotted] node{$a$} (y1)
            (x1) edge[left, below, dotted] node{$b$} (y2)
            (y1) edge[loop left] node{$a,b$} (y1)
            (y2) edge[loop right] node{$a,b$} (y2);
    \end{tikzpicture}
    \hspace{40px}
    $A_2:$ \begin{tikzpicture}[scale=0.7, every node/.style={scale=0.7}]
        \node[state, initial] at (0,0) (x1) {$x_1$};
        \node[state] at (2, 0) (x2) {$x_2$};
        \node[state] at (0, 2) (y1) {$y_1$};
        \node[state, accepting] at (2, 2) (y2) {$y_2$};
        \draw (x1) edge[loop below] node{$a$} (x1)
            (x2) edge[loop below] node{$b$} (x2)
            (x1) edge[bend left, above] node{$b$} (x2)
            (x2) edge[bend left, below] node{$a$} (x1)
            (x1) edge[left, dotted] node{$a,b$} (y1)
            (x2) edge[right, dotted] node{$b$} (y1)
            (x2) edge[right, dotted] node{$a$} (y2)
            (y1) edge[loop left] node{$a,b$} (y1)
            (y2) edge[loop right] node{$a$} (y2)
            (y2) edge[above] node{$b$} (y1);
    \end{tikzpicture}
    \caption{Examples of lasso automata. The dotted arrows are $\sigma$-transitions.}
    \label{fig:lassoAut}
\end{figure}

A state $z$ in a lasso automaton is called \emph{reachable} if there exists a path along $\rho$, $\sigma$ and $\xi$ from the initial state to $z$. If all states in an automaton are reachable, we call it a \emph{reachable automaton}.

A \emph{lasso automaton morphism} is a structure-preserving map between lasso automata. More precisely, given two lasso automata $A_i = (X_i, Y_i, q_i,$ $\rho_i, \sigma_i, \xi_i, F_i)$, for $i \in \{1,2\}$, a lasso automaton morphism is a pair of maps $h = (h^X, h^Y)$ such that $h^X: X_1 \to X_2$ and $h^Y: Y_1 \to Y_2$ satisfy:
\begin{itemize}
    \item $h^X(q_1) = q_2$;
    \item for all $x \in X_1, y \in Y_1, a \in \Sigma$: 
    $h^X(\rho_1(x, a)) = \rho_2(h^X(x), a)$ and \\
    \phantom{xxx} \hfill $h^Y(\sigma_1(x, a)) = \sigma_2(h^X(x), a)$ and $h^Y(\xi_1(y, a)) = \xi_2(h^Y(y), a)$;
    \item  for all $y \in Y_1$: $y \in F_1 \iff h^Y(y) \in F_2$.
\end{itemize}

\begin{remark}
    Consider the endofunctor $G$ on $\Set \times \Set$ defined by $G(X, Y)$ $\coloneqq \langle X^\Sigma \times Y^\Sigma, Y^\Sigma \times 2 \rangle$ on objects \cite{CianciaVenema2012StreamAutomataAreCoalgebras}. Lasso automata are $G$-coalgebras, together with an initial state. Lasso automaton morphisms coincide with initial-state-preserving $G$-coalgebra morphisms.
\end{remark}

In order to capture lasso languages of the form $\{ (u, v) \mid uv^\omega \in L \}$ for an $\omega$-regular language $L$, \cite{CianciaVenema2012StreamAutomataAreCoalgebras} introduces a subclass of lasso automata called $\Omega$-automata. 

\begin{definition}[$\Omega$-automaton {\cite{CianciaVenema2012StreamAutomataAreCoalgebras}}]
    An \emph{$\Omega$-automaton} is a lasso automaton $A = (X, Y, q,\rho, \sigma, \xi, F)$ that satisfies the following two conditions.
    \begin{description}
        \item[Circularity] For all $x \in X, av \in \newords, k > 0$:\\
        $\xi(\sigma(x, a),v) \in F \iff \xi(\sigma(x, a),v(av)^k) \in F$.
        \item[Coherence] For all $x \in X, abv \in \newords$:\\  
        $\xi(\sigma(x, a),bv) \in F \iff \xi(\sigma(\rho(x, a),b),va) \in F$.
    \end{description}
\end{definition}

It is shown in \cite{CianciaVenema2012StreamAutomataAreCoalgebras} that for any $\Omega$-automaton $A$, the language $\Lasso(A)$ is \emph{saturated}. Furthermore, $\Omega$-automata accept precisely the languages of the form $\{ (u, v) \mid uv^\omega \in L \}$ for an $\omega$-regular language $L$.

\begin{example}
    In \Cref{fig:lassoAut}, the automaton $A_2$ is an $\Omega$-automaton, and its  corresponding $\omega$-regular language is $(a+b)^*ba^\omega$. The automaton $A_1$ is circular, but not coherent, because $\xi(\sigma(x,b),a) = y_2 \in F$, $\xi(\sigma(\rho(x, b),a),b) = y_1 \notin F$.
\end{example}

\subsection{Wilke Algebras}
Another approach to characterising the ultimately periodic fragments of $\omega$-regular languages is via recognition by Wilke algebra homomorphisms \cite{Wilke93AlgTheoryForRegLanguagesFinInf} (see also \cite[Section~2.5]{PerrinPin2004InfiniteWords}).

\begin{definition}[Wilke algebra {\cite{Wilke93AlgTheoryForRegLanguagesFinInf}}]\label{def:Wilke-alg}
    A \emph{Wilke algebra} is a two-sorted algebra of the form $W = (\fint W, \inft W, \cdot, \times, (-)^\omega)$, where $\fint W, \inft W$ are sets equipped with the operations:
    \begin{equation*}
        \cdot: \fint W \times \fint W \to \fint W, \qquad
        \times: \fint W \times \inft W \to \inft W,
        \qquad
        (-)^\omega: \fint W \to \inft W,
    \end{equation*}
    satisfying the axioms:
    \begin{align*}
        (s \cdot t) \cdot u &= s \cdot (t \cdot u), &
        s \times (t \times \alpha) &= (s \cdot t) \times \alpha, \\
        (s^n)^\omega &= s^\omega, &
        s \times (t \cdot s)^\omega &= (s \cdot t)^\omega.
    \end{align*}
    for all $s, t \in \fint W$, $\alpha \in \inft W$. The axioms in the second line are called \emph{circularity} and \emph{coherence}, respectively.
\end{definition}
If no confusion arises, we write $W = (\fint W, \inft W)$, i.e., we omit the operations.
A \emph{Wilke algebra homomorphism} between $W_1$ and $W_2$ is a pair $f = (\fint f, \inft f)$ of maps $\fint f: \fint W_1 \to \fint W_2$ and $\inft f: \inft W_1 \to \inft W_2$ that preserves the operations $\cdot$, $\times$ and $(-)^\omega$. That is:
\begin{align*}
    \fint f(s \cdot t) &= \fint f(s) \cdot \fint f(t), &
    \inft f(s \times \alpha) &= \fint f(s) \times \inft f(\alpha), \\
    \inft f(s^{\omega}) &= (\fint f(s))^{\omega}.
\end{align*}
The freely generated Wilke algebra with generators $(\Sigma,\emptyset)$ is $(\newords, \upwords)$, where $\cdot$ is finite-word concatenation, $\times$ is finite-infinite-word concatenation, and $(-)^\omega$ is infinite power. Given a Wilke algebra $W$ and a homomorphism $f: (\newords, \upwords) \to W$, we say $f$ \emph{recognises} a language $L$ of ultimately periodic words if $L = (\inft f)\inv(P)$ for some \emph{recognising subset} $P \subseteq \inft W$, and we write $L = \UP(W, f, P)$. The languages recognised by homomorphisms into \emph{finite} Wilke algebras are precisely the languages of the form $\{ uv^\omega \mid uv^\omega \in L \}$ for an $\omega$-regular $L$.


\section{Algebraic Recognition of Lasso Languages}
\label{sec:algRecognitionLassoLanguages}
In this section, we introduce \emph{lasso semigroups} as generalisations of Wilke algebras, and show that homomorphisms into \emph{finite} lasso semigroups recognise precisely the regular lasso languages. We do this by defining mappings transforming a lasso automaton into a surjective lasso semigroup homomorphism with a recognising set, and vice versa.

\subsection{Lasso Semigroups}
\emph{Lasso semigroups} are obtained by omitting the circularity and coherence axioms of Wilke algebras. We show that the freely generated lasso semigroup over an alphabet $\Sigma$ consists of $\newords$ as its first sort and $\lassos$ as its second sort. This allows us to define recognition of lasso languages via lasso semigroups, analogously to language recognition by Wilke algebras.
\begin{definition}[Lasso semigroup]
\label{def:lassoAlgebra}
    A \emph{lasso semigroup} has the same type as a Wilke algebra $W = (\fint W, \inft W, \cdot, \times, (-)^\omega)$,
    but the circularity and coherence axioms need not be satisfied (cf.~\Cref{def:Wilke-alg}).
    A \emph{lasso semigroup homomorphism} preserves operations in the same way as Wilke algebra homomorphisms.
\end{definition}

From the above definition it follows that Wilke algebras are a full subcategory of lasso semigroups, with their homomorphisms.

\begin{remark}\label{rem:left-action}
    A lasso semigroup is, equivalently, a semigroup $\fint W$ acting on a set $\inft W$ by $\times$, together with a function $(-)^\omega: \fint W \to \inft W$.
\end{remark}

\begin{apxpropositionrep}
    The free lasso semigroup generated by $(\Sigma, \emptyset)$ is (isomorphic to) $(\newords, \lassos)$, where for every $u, v \in \newords$ and $w \in \words$:
    \begin{equation*}
        u \cdot v \coloneqq uv, \qquad
        u \times (w, v) \coloneqq (uw, v), \qquad
        u^\omega \coloneqq (\epsilon, u).
    \end{equation*}
\end{apxpropositionrep}
\begin{proofsketch}
    Suppose $(\fint W, \inft W)$ is a lasso semigroup and $f_0: \Sigma \to \fint W$ is a function. Then $f_0$ can be uniquely extended to a homomorphism $f: (\newords, \lassos) \to (\fint W, \inft W)$ as follows: $\fint f(a_1 \dotsc a_n) \coloneqq f_0(a_1) \cdot \dotsc \cdot f_0(a_n)$ and $\inft f(u,v) \coloneqq \fint f(u) \times \big( \fint f(v) \big)^\omega$. \qed
\end{proofsketch}
\begin{proof}
    The fact that $(\newords, \lassos)$ is a lasso semigroup follows directly from associativity of word concatenation. In order to prove that it is freely generated from $(\Sigma, \emptyset)$, suppose $(\fint W, \inft W)$ is a lasso semigroup and $f_0: \Sigma \to \fint W$ is a function. We show $f_0$ can be uniquely extended to a homomorphism $f: (\newords, \lassos) \to (\fint W, \inft W)$. Given $a_1, \dotsc, a_n \in \newords$ and $(u,v) \in \lassos$, define:
    \begin{align*}
        \fint f(a_1 \dotsc a_n) &\coloneqq f_0(a_1) \cdot \dotsc \cdot f_0(a_n), \\
        \inft f(u,v) &\coloneqq \fint f(u) \times \big( \fint f(v) \big)^\omega.
    \end{align*}
    It is straightforward to verify that $(\fint f, \inft f)$ is a homomorphism and that every homomorphism that extends $f_0$ coincides with $(\fint f, \inft f)$.
    \qed
\end{proof}

Now, analogously to Wilke algebras, given a lasso semigroup homomorphism $(\fint f, \inft f): (\newords, \lassos) \to (\fint W, \inft W)$ and a set $P \subseteq \inft W$, we have that $(\inft f )\inv(P)$ is a lasso language. We say that $(\fint f, \inft f)$ \emph{recognises} $(\inft f )\inv(P)$ via $P$. Note that for every homomorphism $(\fint f, \inft f)$, there exists a surjective homomorphism that recognises the same languages. Indeed, the codomain restriction $(\fint f, \inft f): (\newords, \lassos) \twoheadrightarrow (\Im(\fint f), \Im(\inft f))$ recognises the same languages. Hence in the next definition we only consider surjective homomorphisms.

\begin{definition}[Extended lasso semigroup]
    An \emph{extended lasso semigroup} is a triple $(W, f, P)$ where $W$ is a lasso semigroup, $f: (\newords, \lassos) \twoheadrightarrow W$ is a surjective homomorphism and $P \subseteq \inft W$. We call $(W, f, P)$ \emph{finite} if $W$ is finite. The \emph{lasso language recognised by $(W, f, P)$} is the set $\Lasso(W, f, P)$ $\coloneqq (\inft f)\inv(P)$.
\end{definition}

\begin{remark}
Surjective homomorphisms $f: (\newords, \lassos) \twoheadrightarrow W$ are in 1-1 correspondence with 
congruences on $(\newords, \lassos)$ by taking kernels and quotient maps, respectively.
\end{remark}

In the remainder of this section, we show that the languages recognised by \emph{finite} extended lasso semigroups coincide with the regular lasso languages. Our strategy is to show that: (1) any finite extended lasso semigroup can be transformed into a finite lasso automaton that accepts the \emph{reverse language}; (2) any finite lasso automaton can be transformed into a finite extended lasso semigroup that recognises the \emph{reverse language}. The result then follows from the fact that a language is regular precisely when its reverse is regular (see \cite[Section~8.1]{Cruchten2022TopicsInOmegaAutomata}).

\subsection{From Lasso Semigroups to Lasso Automata}
We define a mapping $\Aut$ that sends an extended lasso semigroup $(W, f, P)$ to a lasso automaton $\Aut(W, f, P)$ accepting $L(W, f, P)^\rev$.

Recall from \Cref{rem:left-action} that a lasso semigroup $(\fint W,\inft W)$ can be seen as a left-action of the semigroup $\fint W$ on the set $\inft W$ via the operation $\times$. The lasso semigroup operations provide a natural way of defining a lasso automaton structure on its two-sorted carrier. This construction is similar to the classic construction of a transition structure from a semigroup $S$ with a semigroup morphism $f\colon \Sigma^+ \to S$ where the transitions are defined 
by $s \stackrel{a}{\longrightarrow} s \cdot f(a)$ \cite{Pin:MathematicalFoundationsOfAutomataTheory}. 
However, since $\times$ is a left-action, we define transitions by multiplying on the left rather than on the right as in the classic construction.    

\begin{definition}[$\Aut$]
\label{def:algAutomaton}
    For an extended lasso semigroup $(W, f, P)$, we define $\Aut(W,f,P)$ as $(\fint W \sqcup \{ \token \}, \inft W, \token, \rho, \sigma, \xi, P)$ where for all $t \in \fint W, \alpha \in \inft W$:
    \begin{itemize}
        \item $\rho(\token, a) \coloneqq \fint f(a) \quad \text{and} \quad \rho(t, a) \coloneqq \fint f(a) \cdot t$;
        \item $\sigma(\token, a) \coloneqq \fint f(a)^\omega \quad \text{and} \quad \sigma(t, a) \coloneqq (\fint f(a) \cdot t)^\omega$;
        \item $\xi(\alpha, a) \coloneqq \fint f(a) \times \alpha$.
    \end{itemize}
\end{definition}

\begin{remark}
    It is clear that if $(W, f, P)$ is finite, then $\Aut(W,f,P)$ is finite.
\end{remark}

Due to defining transitions by multiplying on the left, 
we have (by an easy induction argument) that for all $w \in \Sigma^+$, $\rho(\token,w) = \fint{f}(w^\rev)$.
Similar identities hold for $\sigma$ and $\xi$, and this is essentially the reason why $\Aut(W, f, P)$ accepts the reverse of $\Lasso(W, f, P)$ rather than $\Lasso(W, f, P)$ itself.

\begin{apxpropositionrep}
\label{prop:AutReversesLanguage}
    For every extended lasso semigroup $(W, f, P)$:
    \begin{equation*}
         \Lasso(\Aut(W, f, P)) = \Lasso(W, f, P)^\rev.
    \end{equation*}
\end{apxpropositionrep}
\begin{proofsketch}
    Let $\Aut(W,f,P) = (\fint W \sqcup \{ \token \}, \inft W, \token, \rho, \sigma, \xi, P)$.
    By definition, for all $(u,av) \in \lassos$:
    \[\begin{array}{rcl}
      (u,av) \in \Lasso(\Aut(W,f,P)) & \iff &  \xi(\sigma(\rho(\token, u), a), v) \in P, \text{ and } \\
      (u,av)^\rev \in \Lasso(W, f, P) & \iff & \inft f((u, av)^\rev) \in P 
    \end{array}
    \]
    The proof is completed by showing that:\\ 
    \begin{minipage}[b]{.9\textwidth}
    \begin{equation}\label{eq:algAutomatonReadInReverse}
    \text{ for all } (u,av) \in \lassos: \xi(\sigma(\rho(\token, u), a), v) = \inft f((u, av)^\rev).        
    \end{equation}        
    \end{minipage}
    \qed
\end{proofsketch}
\begin{proof}
    Let $\Aut(W,f,P) = (\fint W \sqcup \{ \token \}, \inft W, \token, \rho, \sigma, \xi, P)$.
    We have for all $(u,av) \in \lassos$:
    \[\begin{array}{rcl}
      (u,av) \in \Lasso(\Aut(W,f,P)) & \iff &  \xi(\sigma(\rho(\token, u), a), v) \in P, \text{ and } \\
      (u,av)^\rev \in \Lasso(W, f, P) & \iff &  \inft f((u, av)^\rev) \in P
    \end{array}
    \]
    The proof is completed by showing that : 
    \begin{equation}\label{eq:Aut-lang-rev}
    \text{for all } (u, av) \in \lassos: \xi(\sigma(\rho(\token, u), a), v) = \inft f((u, av)^\rev).        
    \end{equation}
    
    First, we show $\rho(\token, w) = \fint f(w^\rev)$ whenever $w \in \newords$. We proceed by induction on $|w| \geq 1$. For the base case, if $w = a$ for $a \in \Sigma$, we have:
    \begin{equation*}
        \rho(\token, w) = \rho(\token, a) = \fint f(a) = \fint f(w^\rev).
    \end{equation*}
    For the induction step, if $w = w'a$, for $|w'| \geq 1$, then:
    \begin{equation*}
        \rho(\token, w) = \rho(\rho(\token, w'), a) = \rho(\fint f((w')^\rev), a) = \fint f(a) \cdot \fint f((w')^\rev) = \fint f(w^\rev).
    \end{equation*}
    
    Second, we prove that if $\alpha \in \inft W$ and $w \in \newords$, then $\xi(\alpha, w) = \fint f(w^\rev) \times \alpha$. We proceed by induction on $|w| \geq 1$. For the base case, if $w = a$ for $a \in \Sigma$, we have:
    \begin{equation*}
        \xi(\alpha, w) = \fint f(a) \times \alpha = \fint f (w^\rev) \times \alpha.
    \end{equation*}
    For the induction step, if $w = w'a$ for $|w'| \geq 1$, then:
    \begin{multline*}
        \xi(\alpha, w) = \xi(\xi(\alpha, w'), a) = \xi(\fint f((w')^\rev) \times \alpha, a) = \fint f(a) \times (\fint f((w')^\rev) \times \alpha) = \\ = \fint f(a \cdot (w')^\rev) \times \alpha = \fint f(w^\rev) \times \alpha.
    \end{multline*}
    
    Finally, we show \eqref{eq:Aut-lang-rev}, i.e., that for all $(u, av) \in \lassos$: $\xi(\sigma(\rho(\token, u), a), v) = \inft f(v^\rev, au^\rev)$. 
    We first evaluate $\sigma(\rho(\token, u),a)$. If $u = \epsilon$, then:
    \begin{equation*}
        \sigma(\rho(\token, u),a) = \sigma(\token, a) = \fint f(a)^\omega = \inft f(a^\omega) = \inft f(\epsilon, a) = \inft f(\epsilon, au^\rev).
    \end{equation*}
    Otherwise, we use $\rho(\token, u) = \fint f(u^\rev)$ to obtain:
    \begin{equation*}
        \sigma(\rho(\token, u),a) = \sigma(\fint f(u^\rev), a) = (\fint f(a) \cdot \fint f(u^\rev))^\omega = \inft f(\epsilon, au^\rev)).
    \end{equation*}
    Now we evaluate $\xi(\sigma(\rho(\token, u),a), v)$. If $v = \epsilon$, we have:
    \begin{equation*}
        \xi(\sigma(\rho(\token, u),a), v) = \sigma(\rho(\token, u),a) = \inft f(\epsilon, au^\rev) = \inft f(v^\rev, au^\rev).
    \end{equation*}
    And in case $|v| \geq 1$:
    \begin{multline*}
        \xi(\sigma(\rho(\token, u),a), v) = \xi(\inft f(\epsilon, au^\rev), v) = \fint f(v^\rev) \times \inft f(\epsilon, au^\rev) = \\ = \inft f(v^\rev, au^\rev). \quad \qed
    \end{multline*}
\end{proof}

\subsection{From Lasso Automata to Lasso Semigroups}
We now describe a converse transformation, i.e., a mapping $\Alg$ sending a lasso automaton $A = (X, Y, q, \rho, \sigma, \xi, F)$ to an extended lasso semigroup.
Cruchten \cite[Ch.~5]{Cruchten2022TopicsInOmegaAutomata} gives a construction of a Wilke algebra from an $\Omega$-automaton
which can be seen as a generalisation of the classic construction of a transition semigroup from a finite automaton.
In the construction in ibid., an element of the algebra represents paths in the automaton corresponding to a word. Our construction $\Alg$ is a variation of Cruchten's idea, with the crucial difference that here paths are reversed. The choice of $\Alg$ is justified in \Cref{sec:adjunctionLassoAutLassoAlg}, where we show that $\Alg$ is the (unique) right adjoint of $\Aut$ (\Cref{prop:AutAlgAdjunction}).

As the carrier of the algebra we take $U_A \coloneqq (X^X \times Y^X \times Y^Y, Y)$. That is, elements of $\fint U_A$ are triples $(\alpha, \beta, \gamma)$, where $\alpha$ encodes (the endpoints of) $\rho$-paths, $\beta$ encodes $\rho$-paths with a single $\sigma$-transition at the end, and $\gamma$ encodes $\xi$-paths. Elements $y \in \inft U$ represent the state reached after reading some lasso \emph{in reverse}, starting from $q$. Before defining the operations on $U_A$, it is insightful to see what the desired homomorphism $f_A: (\newords, \lassos) \to U_A$ is:
\begin{align}
    \fint f_A(av) &= (\lambda x . \rho(x, v^\rev a), \lambda x. \sigma(\rho(x, v^\rev), a), \lambda y . \xi(y, v^\rev a)), \label{eq:fFin} \\
    \inft f_A(u, av) &= \xi(\sigma(\rho(q, v^\rev ),a), u^\rev). \label{eq:fInf}
\end{align}
In fact, in defining the operations on $U_A$, we are guided by the goal of ensuring $f_A$ becomes a homomorphism. The fact that our construction reverses the language will follow from the form of $f_A$.

\begin{definition}
\label{def:automatonAlgebraAmbientCarrier}
    Let $A = (X, Y, q, \rho, \sigma, \xi, F)$ be a lasso automaton. Define the algebraic structure $U_A \coloneqq (X^X \times Y^X \times Y^Y, Y)$ with the following operations:
    \begin{align*}
        (\alpha_1, \beta_1, \gamma_1) \cdot (\alpha_2, \beta_2, \gamma_2) &\coloneqq (\alpha_1 \alpha_2, \beta_1 \alpha_2, \gamma_1 \gamma_2), \\
        (\alpha, \beta, \gamma)^\omega &\coloneqq \beta(q), \\
        (\alpha, \beta, \gamma) \times y &\coloneqq \gamma(y),
    \end{align*}
    for each $\alpha_i \in X^X$, $\beta_i \in Y^X$, $\gamma_i \in Y^Y$, $y \in Y$ (here $\alpha_1 \alpha_2$ denotes $\alpha_1 \circ \alpha_2$).
\end{definition}

\begin{apxpropositionrep}
    The structure defined in \Cref{def:automatonAlgebraAmbientCarrier} is a lasso semigroup.
\end{apxpropositionrep}
\begin{proof}
     For all $(\alpha_i, \beta_i, \gamma_i) \in \fint U_A$, where $i \in \{1,2,3\}$, and $y \in \inft U_A$:
     \begin{align*}
         \big((\alpha_1, \beta_1, \gamma_1) \cdot (\alpha_2, \beta_2, \gamma_2)\big) \cdot (\alpha_3, \beta_3, \gamma_3) &= (\alpha_1 \alpha_2, \beta_1 \alpha_2, \gamma_1 \gamma_2) \cdot (\alpha_3, \beta_3, \gamma_3)\\ &= (\alpha_1 \alpha_2 \alpha_3, \beta_1 \alpha_2 \alpha_3, \gamma_1 \gamma_2 \gamma_3) \\ &= (\alpha_1, \beta_1, \gamma_1) \cdot (\alpha_2 \alpha_3, \beta_2 \alpha_3, \gamma_2 \gamma_3) \\ &= (\alpha_1, \beta_1, \gamma_1) \cdot \big((\alpha_2, \beta_2, \gamma_2) \cdot (\alpha_3, \beta_3, \gamma_3)\big) \\
         \big((\alpha_1, \beta_1, \gamma_1) \cdot (\alpha_2, \beta_2, \gamma_2)\big) \times y &= (\alpha_1 \alpha_2, \beta_1 \alpha_2, \gamma_1 \gamma_2) \times y \\ &= \gamma_1 \gamma_2(y) \\ &= (\alpha_1, \beta_1, \gamma_1) \times \gamma_2(y) \\ &= (\alpha_1, \beta_1, \gamma_1) \times \big((\alpha_2, \beta_2, \gamma_2) \times y\big). \hspace{14px} \qed
     \end{align*}
\end{proof}

\begin{apxpropositionrep}
    Let $A = (X, Y, q, \rho, \sigma, \xi, F)$ be a lasso automaton and $f_A:(\newords,$ $\lassos) \to U_A$ be defined by \Cref{eq:fFin,eq:fInf}. Then $f_A$ is a lasso semigroup homomorphism.
\end{apxpropositionrep}
\begin{proof}
    For all $au, bv \in \newords$ and $w \in \words$:
    \begin{align*}
        \fint f_A(au) \cdot \fint f_A(bv) &=
        (\lambda x.\rho(x, u^\rev a), \lambda x. \sigma(\rho(x, u^\rev), a), \lambda y. \xi(y, u^\rev a)) \;\cdot \\
        &\hspace{13px}(\lambda x.\rho(x, v^\rev b), \lambda x. \sigma(\rho(x, v^\rev), b), \lambda y. \xi(y, v^\rev b)) \\
        &= (\lambda x.\rho(x, v^\rev b u^\rev a), \lambda x. \sigma(\rho(x, v^\rev b u^\rev), a), \lambda y. \xi(y, v^\rev b u^\rev a)) \\
        &= \fint f_A(aubv) = \fint f_A(au \cdot bv), \\
        \fint f_A(au)^\omega &= (\dotsc, \lambda x. \sigma(\rho(x, u^\rev),a), \dotsc)^\omega = \sigma(\rho(q, u^\rev),a) \\
        &= \inft f_A(\epsilon, au) = \inft f_A((au)^\omega), \\
        \fint f_A(au) \times \inft f_A(w, bv) &= (\dotsc, \dotsc, \lambda y. \xi(y, u^\rev a)) \times \xi(\sigma(\rho(q,v^r),b),w^\rev) \\
        &= \xi(\sigma(\rho(q, v^\rev), b),w^\rev u^\rev a)
        = \inft f_A(auw, bv) \\
        &= \inft f_A(au \times (w, bv)).
        \hspace{169px} \qed
    \end{align*}
\end{proof}

Note that $f_A$ is not surjective, but we can define the desired extended lasso semigroup by taking the image of $f_A$.

\begin{definition}[$\Alg$]
    Given a lasso automaton $A = (X, Y, q, \rho, \sigma, \xi, F)$, we define $\Alg(A)$ $\coloneqq (W_A, f_A, F)$, where $W_A$ is the image of $f_A$ in $U_A$.
\end{definition}

\begin{remark}
    It follows immediately that if $A$ is finite, then $\Alg(A)$ is finite.
\end{remark}

\begin{proposition}
\label{prop:AlgReversesLanguage}
    For every lasso automaton $A = (X, Y, q, \rho, \sigma, \xi, F)$:
    \begin{equation*}
         \Lasso(\Alg(A)) = \Lasso(A)^\rev.
    \end{equation*}
\end{proposition}
\begin{proof}
    Suppose $\Alg(A) = (W, f, P)$. We have that $\Lasso(W, f, P)$ consists of all lassos $(u, av)$ such that $\inft f(u,av) \in P$. By \Cref{eq:fInf}, this is equivalent to $\xi(\sigma(\rho(s, v^\rev), a), u^\rev) \in P = F$, i.e., $(v^\rev, au^\rev) \in \Lasso(A)$. Hence $\Lasso(W, f, P)$ $=\Lasso(A)^\rev$.
    \qed
\end{proof}

\subsection{Finite Lasso Semigroups Recognise Regular Lasso Languages}

From \Cref{prop:AutReversesLanguage} and \Cref{prop:AlgReversesLanguage} it follows that the languages recognised by finite extended lasso semigroups are the reverse of regular lasso languages. In order to conclude that finite extended lasso semigroups recognise regular lasso languages, it remains to show that $L$ is regular if and only if $L^\rev$ is regular. This follows from the fact that, analogously to DFAs, every lasso automaton can be reversed. The construction is described in \cite[Section~8.1]{Cruchten2022TopicsInOmegaAutomata}. States in the reversed automaton are sets of states of the original automaton, while transitions correspond to taking preimages of the original transition functions. We include the definition here, since it will be used in \Cref{sec:adjunctionLassoAutLassoAlg,sec:adjunctionOmegaAutomata}.

\begin{definition}[Reverse lasso automaton {\cite[Def.~8.17]{Cruchten2022TopicsInOmegaAutomata}}]
\label{def:reverseLassoAutomaton}
    Let $A = (X, Y, q,$ $\rho,\sigma, \xi,F)$ be a lasso automaton. Define the reverse automaton $\Rev(A) \coloneqq (2^Y, 2^X,$ $F, \hat \xi, \hat \sigma, \hat \rho, \{S \mid q \in S\})$, where, for $\delta \in \{ \rho, \sigma, \xi \}$, $\hat \delta$ is defined as:
    \begin{equation*}
        \hat \delta(S, a) \coloneqq \{ z \mid \delta(z, a) \in S \}.
    \end{equation*}
\end{definition}

\begin{proposition}[{\cite[Prop.~8.22]{Cruchten2022TopicsInOmegaAutomata}}]
    Let $A$ be a lasso automaton. Then $\Lasso(\Rev$ $(A)) = \Lasso(A)^\rev$.
\end{proposition}

We can now state our algebraic characterisation of regular lasso languages.
\begin{theorem}
\label{thm:regularLanguageCharViaLassoAlg}
    A lasso language $L$ is recognised by a finite extended lasso semigroup if and only if $L$ is regular.
\end{theorem}
\begin{proof}
    Suppose $L = \Lasso(W, f, P)$ for some finite extended lasso semigroup $(W,f,$ $P)$. Then $L = \Lasso(\Rev(\Aut(W, f, P))$, where $\Rev(\Aut(W, f, P))$ is a finite lasso automaton, thus $L$ is regular. Conversely, if $L = \Lasso(A)$ for some finite lasso automaton, then $L = \Lasso(\Alg(\Rev(A)))$, where $\Alg(\Rev(A))$ is a finite extended lasso semigroup.
    \qed
\end{proof}

Noting that $(W,f,P)$ recognises $L$ iff $(W,f,\inft{W}\setminus P)$ recognises the complement of $L$, 
\Cref{thm:regularLanguageCharViaLassoAlg} implies 
that regular lasso languages are closed under complementation. This was already proved using automata in \cite{CianciaVenema2012StreamAutomataAreCoalgebras}, but the algebraic argument is immediate.\footnote{To prove closure under union and intersection, we would additionally need to consider limits of lasso semigroups.}

\begin{corollary}
    Regular lasso languages are closed under complementation.
\end{corollary}


\section{Dual Adjunction Between Lasso Automata and Lasso Semigroups}
\label{sec:adjunctionLassoAutLassoAlg}
In the last section we introduced the mappings $\Aut$ and $\Alg$ as tools for characterising language recognition by finite extended lasso semigroups. In this section, we show that $\Aut$ and $\Alg$ also reveal the categorical relationship between the \emph{category of lasso automata} and the \emph{category of extended lasso semigroups}. More precisely, we show that $\Aut$ and $\Alg$ can be extended to a pair of adjoint functors. By composing this adjunction with the adjunction $\Rev \dashv \op\Rev$ proven in \cite[Section~8.1]{Cruchten2022TopicsInOmegaAutomata}, we arrive at a language-preserving dual adjunction between extended lasso semigroups and lasso automata. See Diagram \labelcref{diag:lassoAutAdjInformal}.

\subsection{Categories of Lasso Automata and Lasso Semigroups}
A natural notion of a morphism between extended lasso semigroups is a homomorphism that preserves the quotient structure and the recognising subset.

\begin{definition}[Category of extended lasso semigroups]
    Given two extended lasso semigroups $(W_i, f_i, P_i)$, an \emph{extended lasso semigroup morphism} $g: (W_1, f_1, P_1) \to (W_2, f_2, P_2)$ is a homomorphism $g: W_1 \to W_2$ such that $g \circ f_1 = f_2$ and $\alpha \in P_1 \iff \inft g(\alpha) \in P_2$, for all $\alpha \in \inft W_1$. We write $\ELAlg$ for the category of extended lasso semigroups and their morphisms.
\end{definition}

On the automaton side, we use the standard notion of automaton morphism (see \Cref{sec:preliminaries}). Apart from the category of all lasso automata, we define its full subcategory of reachable lasso automata. A restriction to reachable automata is necessary in \Cref{prop:AutAlgFunctorial} for ensuring $\Alg$ is functorial.

\begin{definition}[Categories of lasso automata]
Let $\LAut$ denote the category of lasso automata and lasso automata morphisms. Let $\RLAut$ denote the full subcategory of $\LAut$ of all reachable lasso automata.
\end{definition}

It follows from surjectivity of $f_1$ that there is at most one extended lasso semigroup morphism $g: (W_1, f_1, P_1) \to (W_2, f_2, P_2)$. Moreover, observe that if $A$ is a reachable lasso automaton, then there exists at most one morphism with domain $A$. That is:

\begin{lemma}
\label{lem:posetalCategories}
    $\RLAut$ and $\ELAlg$ are posetal categories.
\end{lemma}

\subsection{Functoriality of \texorpdfstring{$\Aut$}{Aut} and \texorpdfstring{$\Alg$}{Alg}}
We begin with an example showing that $\Alg$ \emph{cannot be extended} to a functor $\LAut \to \ELAlg$.

\begin{example}
    Consider the lasso automata $A \coloneqq (\{x\}, \{y\}, x, \rho, \sigma, \xi, \emptyset)$ (where $\rho$, $\sigma$ and $\xi$ are uniquely determined by their types) and $A' = (\{x_1', x_2'\}, \{y'\}, x_1', \rho', \sigma',$ $\xi', \emptyset)$ with $\rho'(x_1', a) = \rho'(x_1', b) = \rho(x_2', a) = x_1'$, $\rho'(x_2', b) = x_2'$. The map $h = (h^X, h^Y)$ with $h^X(x) = x_1'$ and $h^Y(y) = y'$ is a lasso automaton morphism. However, there is no map from $\Alg(A) \coloneqq (W, f, P)$ to $\Alg(A') \coloneqq (W', f', P')$, because $\fint f(a) = \langle \{ x \mapsto x \}, \{ x \mapsto y \}, \{ y \mapsto y \} \rangle  = \fint f(b)$, but $\fint{(f')}(a) = \langle \{ x_1' \mapsto x_1', x_2' \mapsto x_1' \}, \dotsb, \dotsb \rangle \neq \langle \{ x_1' \mapsto x_1', x_2' \mapsto x_2' \}, \dotsb, \dotsb \} \rangle = \fint{(f')}(b)$.
\end{example}

Hence in order to obtain a functor $\Alg$, we need to restrict the domain $\LAut$. In the example above, the automaton $A'$ was not reachable, which gives us the idea to restrict the domain to $\RLAut$. Moreover, the next lemma shows that the codomain of $\Aut$ can also be restricted to $\RLAut$.

\begin{apxlemmarep}
\label{lem:AutSendsToReachable}
    Let $(W, f, P)$ be an extended lasso semigroup. Then $\Aut(W, f, P)$ is reachable.
\end{apxlemmarep}
\begin{proof}
    Suppose $Aut(W, f, P) = (X, Y, q, \rho, \sigma, \xi, F)$. Let $x \in X$, we show that $x = \rho(q, v)$ for some $v$ in $\words$. If $x = q$, then $x = \rho(q, \epsilon)$. Otherwise, $x \in \fint W$, so $x = \fint f(w)$ for some $w$ in $\newords$. If $w = a_1 \dotsc a_n$ for $a_1, \dotsc, a_n$ in $\Sigma$, then:
    \begin{equation*}
        x = \fint f(a_1 \dotsc a_n) = \fint f(a_1) \dotsc \fint f(a_n) = \rho(q, a_n \dotsc a_1).
    \end{equation*}
    
    Let $y \in Y$, we show that $y = \xi(\sigma(\rho(q, u), a), v)$ for some $(u,av) \in \lassos$. By surjectivity of $f$, we know $y = \inft f(\alpha)$ for some $\alpha$ in $\lassos$. Suppose $\alpha = (a_1 \dotsc a_m, b_1 \dotsc b_n)$ for some $a_1, \dotsc, a_m, b_1, \dotsc, b_m \in \Sigma$, where $m \geq 0$, $n > 0$. Then:
    \begin{equation*}
        y = \fint f(a_1) \times (\dotsc \fint f(a_{m-1}) \times (\fint f(a_m) \times \big(\fint f(b_1) \dotsc \fint f(b_n)\big)^\omega) \dotsc),
    \end{equation*}
    thus $y = \xi(\sigma(\rho(q, b_n \dotsc,b_2), b_1), a_m \dotsc a_1)$.
    \qed
\end{proof}

Since $\ELAlg$ is a posetal category (\Cref{lem:posetalCategories}), for any lasso automaton morphism $h: A_1 \to A_2$, there exists at most one candidate for $\Alg(h)$. Thus, in order to show functoriality of $\Alg$, it suffices to prove that such a candidate exists. Likewise for functoriality of $\Aut$.

We reduce existence of morphisms in $\ELAlg$ and $\RLAut$ to comparing certain equivalence relations on $\newords$ and $\lassos$.

\begin{definition}
\label{def:automatonLanguageCongruence}
    Let $A = (X, Y, q, \rho, \sigma, \xi, F)$ be a lasso automaton. We write:
    \begin{equation*}
        \chi_A(ua) \coloneqq (\lambda x . \rho(x, ua), \lambda x. \sigma(\rho(x, u), a), \lambda y . \xi(y, ua)).
    \end{equation*}
    Define the pair ${\autcong_A} = (\fincong_A, \infcong_A)$ of equivalence relations $\fincong_A$ on $\newords$ and $\infcong_A$ on $\lassos$ by:
    \begin{align*}
        u_1 \fincong_A u_2 &\iff \chi_A(u_1) = \chi_A(u_2) \\
        (v_1, a_1u_1) \infcong_A (v_2, a_2u_2) &\iff \xi(\sigma(\rho(q,v_1),a_1),u_1) = \xi(\sigma(\rho(q,v_2),a_2),u_2).
    \end{align*}
    We say that $\autcong_{A_1}$ \emph{refines} $\autcong_{A_2}$ if ${\fincong_{A_1}} \subseteq {\fincong_{A_2}}$ and ${\infcong_{A_1}} \subseteq {\infcong_{A_2}}$. Define $\autcong_A^\rev$ by:
    \begin{align*}
        u_1 \autcong_A^\rev u_2 &\iff u_1^\rev \autcong_A u_2^\rev,\\
        (v_1, a_1u_1) \autcong_A^\rev (v_2, a_2u_2) &\iff (v_1, a_1u_1)^\rev \autcong_A (v_2, a_2u_2)^\rev.
    \end{align*}
\end{definition}

Compare the definition of $\autcong_A$ to \Cref{eq:fFin,eq:fInf}. We have $u_1 \autcong_A u_2 \iff \fint f_A(u_1^\rev) = \fint f_A(u_2^\rev)$, and $(v_1, u_1) \autcong_A (v_2, u_2) \iff \inft f_A((v_1, u_1)^\rev) = \inft f_A((v_2, u_2)^\rev)$.
Furthermore, note that $\autcong_A$ resembles the relation from \cite[Def.~6.3]{Cruchten2022TopicsInOmegaAutomata} used for deriving a Myhill-Nerode theorem \cite[Th.~6.13]{Cruchten2022TopicsInOmegaAutomata} for $\Omega$-automata. There $u_1$ and $u_2$ are identified if $\rho(q, u_1) = \rho(q, u_2)$. Our $\fincong_A$ is more restrictive, 
since it considers all types of transitions $\rho, \sigma, \xi$, and all starting states.

\begin{definition}
\label{def:lassoSemLanguageCongruence}
    Let $(W,f,P)$ be an extended lasso semigroup. Define the pair ${\algsim_W} = (\finsim_W, \infsim_W)$ of equivalence relations $\finsim_W$ on $\newords$ and $\infsim_A$ on $\lassos$ where $\finsim_W$ is the kernel of $\fint f$ and $\infsim_W$ is the kernel of $\inft f$. Refinement and $\algsim_W^\rev$ are defined analogously to \Cref{def:automatonLanguageCongruence}.
\end{definition}

\begin{apxlemmarep}
\label{lem:morphismIfCongruenceRefinement}
    \begin{enumerate}
        \item Let $A_1$ and $A_2$ be reachable lasso automata. There exists an automaton morphism $h: A_1 \to A_2$ if and only if $\autcong_{A_1}$ refines $\autcong_{A_2}$ and $\Lasso(A_1) = \Lasso(A_2)$.
        \item Let $(W_1, f_1, P_1)$ and $(W_2, f_2, P_2)$ be extended lasso semigroups. There exists an extended lasso semigroup morphism $g: (W_1, f_1, P_1) \to (W_2, f_2, P_2)$ if and only if $\algsim_{W_1}$ refines $\algsim_{W_2}$ and $\Lasso(W_1,f_1,P_1) = \Lasso(W_2,f_2,P_2)$.
    \end{enumerate}
\end{apxlemmarep}
\begin{proof}
    \begin{enumerate}
        \item ($\Rightarrow$) Suppose $h: A_1 \to A_2$ is a lasso automaton morphism, where $A_i = (X_i, Y_i, q_i, \rho_i, \sigma_i, \xi_i, F_i)$. We have:
        \begin{multline*}
            u_1 \fincong_{A_1} u_2 \iff \chi_{A_1}(u_1) = \chi_{A_1}(u_2) \implies \chi_{A_2}(u_1) = \chi_{A_2}(u_2) \\
            \iff
            u_1 \fincong_{A_2} u_2.
        \end{multline*}
        Analogously, $(v_1, a_1u_1) \infcong_{A_1} (v_2, a_2u_2)$ implies $(v_1, a_1u_1) \infcong_{A_2} (v_2, a_2u_2)$. Hence $\autcong_{A_1}$ refines $\autcong_{A_2}$. Lastly:
        \begin{multline*}
            (v,av) \in \Lasso(A_1) \iff \xi_1(\sigma_1(\rho_1(q_1, v),a),u) \in F_1 \iff \\
            h^Y(\xi_1(\sigma_1(\rho_1(q_1, v),a),u)) \in F_2 \iff \xi_2(\sigma_2(\rho_2(q_2, v),a),u) \in F_2 \iff\\
            (v,au) \in \Lasso(A_2).
        \end{multline*}

        ($\Leftarrow$) Suppose $\autcong_{A_1}$ refines $\autcong_{A_2}$ and $\Lasso(A_1) = \Lasso(A_2)$. Define $h = (h^X, h^Y)$ as follows:
        \begin{align*}
            h^X(x) &\coloneqq \rho_2(q_2, u), & &\text{for some $u$ such that $x = \rho_1(q_1, u)$}, \\
            h^Y(y) &\coloneqq \xi_2(\sigma_2(\rho_2(q_2, v), a),u), &  &\text{for some $(v,au)$ such that} \\
            & & &\text{$y = \xi_1(\sigma_1(\rho_1(q_1, v), a),u)$}.
        \end{align*}
        Note that ${\fincong_{A_1}} \subseteq {\fincong_{A_2}}$ ensures that $h^X$ is well-defined, and ${\infcong_{A_1}} \subseteq {\infcong_{A_2}}$ ensures that $h^Y$ is well-defined. Totality of $h^X$ and $h^Y$ follows from reachability of $A_1$. Preservation of the initial state and transitions follows from the definition of $h$. Preservation of final states follows from $\Lasso(A_1) = \Lasso(A_2)$. Hence $h$ is a lasso automaton morphism.
        \item ($\Rightarrow$) Suppose $g: (W_1,f_1,P_1) \to (W_2,f_2,P_2)$ is an extended lasso semigroup morphism. Then:
        \begin{multline*}
            u_1 \finsim_{W_1} u_2 \iff \fint f_1(u_1) = \fint f_1(u_2) \implies  g \fint f_1(u_1) = g \fint f_1(u_2) \iff \\
            \fint f_2(u_1) = \fint f_2(u_2) \iff u_1 \finsim_{W_2} u_2.
        \end{multline*}
        Analogously, $(v_1, a_1u_1) \infsim_{W_1} (v_2, a_2u_2)$ implies $(v_1, a_1u_1) \infsim_{W_2} (v_2, a_2u_2)$. Hence $\algsim_{W_1}$ refines $\algsim_{W_2}$. Moreover:
        \begin{multline*}
            \Lasso(W_1,f_1,P_1) = (\inft f_1)\inv(P_1) = (\inft f_1)\inv(g\inv(P_2)) = (\inft f_2)\inv(P_2) = \\
            = \Lasso(W_2,f_2,P_2).
        \end{multline*}

        ($\Leftarrow$) Suppose $\algsim_{W_1}$ refines $\algsim_{W_2}$ and $\Lasso(W_1,f_1,P_1) = \Lasso(W_2,f_2,P_2)$. Define $g$ as follows:
        \begin{align*}
            \fint g(s) &\coloneqq \fint f_2(u), & &\text{for some $u$ such that $s = \fint f_1(u)$}, \\
            \inft g(\alpha) &\coloneqq \inft f_2(v,u), &  &\text{for some $(v,u)$ such that $\alpha = \inft f_2(v,u)$}.
        \end{align*}
        Observe that $\fint g$ is well-defined since $\finsim_{W_1} \subseteq \finsim_{W_2}$, and $\inft g$ is well-defined since $\infsim_{W_1} \subseteq \infsim_{W_2}$. Totality of $g$ follows from surjectivity of $f_1$. The property $f_2 = gf_1$ follows from the definition of $g$. The property $\alpha \in P_1 \iff \inft g(\alpha) \in P_2$ follows from $\Lasso(W_1,f_1,P_1) = \Lasso(W_2,f_2,P_2)$. \qed
    \end{enumerate}
\end{proof}

\begin{apxlemmarep}
\label{lem:AutAlgCongruences}
    Let $A$ be a lasso automaton and $(W,f,P)$ be an extended lasso semigroup. Then ${\algsim_{\Alg(A)}} = {\autcong_A^\rev}$ and ${\autcong_{\Aut(W,f,P)}} = {\algsim_W^\rev}$.
\end{apxlemmarep}
\begin{proof}
    Let $A = (X, Y, q, \rho, \sigma, \xi, F)$ and $\Alg(A) = (W_A, f_A, P_A)$. Using \Cref{eq:fFin,eq:fInf}:
    \begin{align*}
        &u_1 \finsim_{\Alg(A)} u_2 \iff \fint f_A(u_1) = \fint f_A(u_2) \iff  \chi_A(u_1^\rev) = \chi_A(u_2^\rev) \\
        & \hspace{58px} \iff u_1^\rev \fincong_A u_2^\rev, \\
        &(v_1,a_1u_1) \infsim_{\Alg(A)} (v_2,a_2v_2) \iff \inft f_A(v_1,a_1u_1) = \inft f_A(v_2,a_2u_2) \iff \\
        &\hspace{80px}\xi(\sigma(\rho(q, u_1^\rev),a_1),v_1^\rev) = \xi(\sigma(\rho(q, u_2^\rev),a_2),v_2^\rev) \iff \\
        &\hspace{180px}(v_1,a_1u_1)^\rev \infcong_A (v_2,a_2u_2)^\rev.
    \end{align*}
    Hence ${\algsim_{\Alg(A)}} = {\autcong_A^\rev}$.

    Let $(W,f,P)$ be an extended lasso semigroup and $\Aut(W,f,P) = (X_W,Y_W,$ $q_W,\rho_W,\sigma_W,\xi_W,F_W)$. Using \Cref{def:algAutomaton} and equation \eqref{eq:algAutomatonReadInReverse} on page~\pageref{eq:algAutomatonReadInReverse}:
    \begin{align*}
        &u_1 \fincong_{\Aut(W,f,P)} u_2 \iff \chi_{\Aut(W,f,P)}(u_1) = \chi_{\Aut(W,f,P)}(u_2) \\
        &\hspace{120px}\iff \fint f(u_1^\rev) = \fint f(u_2^\rev) \iff u_1^\rev \finsim_W u_2^\rev, \\
        &(v_1,a_1u_1) \infcong_{\Aut(W,f,P)} (v_2,a_2u_2) \\
        &\hspace{10px}\iff \xi_W(\sigma_W(\rho_W(q_W, v_1),a_1),u_1) = \xi_W(\sigma_W(\rho_W(q_W, v_2),a_2),u_2) \\
        &\hspace{10px}\iff \inft f((v_1,a_1u_1)^\rev) = \inft f((v_2,a_2u_2)^\rev) \iff (v_1,a_1u_1)^\rev \infsim_W (v_2,a_2u_2)^\rev.
    \end{align*}
    Hence ${\autcong_{\Aut(W,f,P)}} = {\algsim_W^\rev}$.
\end{proof}

\begin{proposition}
\label{prop:AutAlgFunctorial}
    The mappings $\Alg$ and $\Aut$ can be extended uniquely to functors $\Alg: \RLAut \to \ELAlg$ and $\Aut: \ELAlg \to \RLAut$.
\end{proposition}
\begin{proof}
    First, we prove functoriality of $\Alg$. Let $h = (h^X, h^Y): A_1 \to A_2$ be a lasso automaton morphism. Because of \Cref{lem:posetalCategories}, it suffices to show that there exists a morphism $g: \Alg(A_1) \to \Alg(A_2)$.
    By \Cref{lem:morphismIfCongruenceRefinement}, $\autcong_{A_1}$ refines $\autcong_{A_2}$ and $\Lasso(A_1) = \Lasso(A_2)$. By \Cref{lem:AutAlgCongruences}, $\algsim_{\Alg(A_1)}^\rev$ refines $\algsim_{\Alg(A_2)}^\rev$, so $\algsim_{\Alg(A_1)}$ refines $\algsim_{\Alg(A_2)}$. By \Cref{prop:AlgReversesLanguage}, $\Lasso(\Alg(A_1)) = \Lasso(A_1)^\rev = \Lasso(A_2)^\rev = \Lasso(\Alg(A_2))$. By \Cref{lem:morphismIfCongruenceRefinement} again, there exists a morphism $g: \Alg(A_1) \to \Alg(A_2)$.
    Functoriality of $\Aut$ follows analogously.
\end{proof}

\subsection{Lasso Adjunction}
Below we prove the dual adjunction between lasso automata and extended lasso semigroups. It is obtained as the composition of three simpler adjunctions (cf.~Diagram \labelcref{eq:lassoAutomatonAdjunctionDiagram}). We start with the adjunction between reachable lasso automata and extended lasso semigroups $\Aut \dashv \Alg$. It is the key technical result of this paper. In the proof, we work with the definition of adjunctions in terms of hom-sets, cf. \cite[Section~9.2]{Awodey2006CategoryTheory}.

\begin{proposition}
\label{prop:AutAlgAdjunction}
    There exists an adjunction $\Aut \dashv \Alg: \ELAlg \to \RLAut$.
\end{proposition}
\begin{proof}
    Let $A = (X, Y, q, \rho, \sigma, \xi, F)$ be an arbitrary reachable lasso automaton, $(W, f, P)$ an arbitrary extended lasso semigroup. Since $\RLAut$ and $\ELAlg$ are posetal categories, it suffices to show that there exists a morphism $g: (W,f,P) \to \Alg(A)$ if and only if there exists a morphism $h: \Aut(W, f, P) \to A$. Suppose there exists $g: (W,f,P) \to \Alg(A)$. By \Cref{lem:morphismIfCongruenceRefinement}, $\algsim_W$ refines $\algsim_{\Alg(A)}$ and $\Lasso(\Alg(A)) = \Lasso(W,f,P)$. By \Cref{lem:AutAlgCongruences}, $\algsim_W$ refines $\autcong_A^\rev$. By \Cref{lem:AutAlgCongruences} again, $\autcong_{\Aut(W,f,P)}^\rev$ refines $\autcong_A^\rev$, so $\autcong_{\Aut(W,f,P)}$ refines $\autcong_A$. By \Cref{prop:AutReversesLanguage} and \Cref{prop:AlgReversesLanguage}, $\Lasso(A) = \Lasso(\Alg(A))^\rev = \Lasso(W,f,P)^\rev = \Lasso(\Aut(W,$ $f,P))$. By \Cref{lem:morphismIfCongruenceRefinement} again, there exists $h: \Aut(W,f,P) \to A$. The other direction is analogous.
    \qed
\end{proof}

Although $\Aut \dashv \Alg$ reveals a relationship between lasso automata and extended lasso semigroups, it leaves more to be desired. Concretely, we look for an adjunction: (1) that is also defined for non-reachable automata, and (2) whose constituent functors preserve the accepted language. Language-preservation enables specialising the adjunction to $\Omega$-automata and Wilke algebras in \Cref{sec:adjunctionOmegaAutomata}.

In order to handle the first requirement, we give an adjunction between $\RLAut$ and $\LAut$. It is analogous to a similar adjunction between reachable DFAs and all DFAs \cite[Section~9.4]{BonchiEtAl2014AlgebraCoalgebraDualityInBrzozowski}. In one direction, we have an inclusion functor $\Incl: \RLAut \to \LAut$. For the the other direction, there exists a functor $\Reach: \LAut \to \RLAut$ mapping an automaton $A$ to its reachable part $\Reach(A)$. Moreover, $\Reach$ maps an automaton morphism to its restriction to reachable states.

\begin{apxpropositionrep}
    There exists an adjunction $\Incl \dashv \Reach: \RLAut \to \LAut$.
\end{apxpropositionrep}
\begin{proofsketch}
    Every morphism in $\Hom(A, \Reach(B))$ can be mapped bijectively to a morphism in $\Hom(\Incl(A), B)$ by expanding its codomain.
\end{proofsketch}
\begin{proof}
    Let $A \in \RLAut, B \in \LAut$. We show the existence of a natural isomorphism:
    \begin{equation*}
        \phi_{A,B}: \Hom(A, \Reach(B)) \to \Hom(\Incl(A), B).
    \end{equation*}
    Let $f: A \to \Reach(B)$ be a morphism in $\RLAut$. Since $\Incl(A) = A$ and $\Reach(B) \subseteq B$, $f$ can also be seen as a morphism between $\Incl(A)$ and $B$ by expanding its codomain. Define $\phi_{A,B}(f)$ to be the result of this codomain change on $f$. Now $\phi_{A,B}$ is clearly injective. Moreover, by reachability of $\Incl(A)$ and the fact that morphisms send reachable states to reachable states, the range of every morphism between $\Incl(A)$ and $B$ is contained in $\Reach(B)$. Hence $\phi_{A,B}$ is also surjective. Naturality of the isomorphism is simple to verify. \qed
\end{proof}

In order to handle the second requirement, we recall from \cite[Section~8.1]{Cruchten2022TopicsInOmegaAutomata} that $\Rev$ from \Cref{def:reverseLassoAutomaton} can be extended to a functor which is its own dual adjoint. That is, \cite[Def.~8.23]{Cruchten2022TopicsInOmegaAutomata} extends $\Rev$ to a functor by defining it on morphisms as $\Rev(h^X,h^Y) = ((h^X)\inv, (h^Y)\inv)$. Then \cite[Cor.~8.24]{Cruchten2022TopicsInOmegaAutomata} states that there is an adjunction $\Rev \dashv \op\Rev: \LAut \to \op\LAut$.

Now we are ready to collect all adjunctions into the main result of this section.

\begin{theorem}
\label{thm:lassoAutomataAdjunctionChain}
    The functors $\Rev \circ \Incl \circ \Aut$ and $\Aut \circ \Reach \circ \op\Rev$ are language-preserving adjoints, with $\Rev \circ \Incl \circ \Aut \dashv \Alg \circ \Reach \circ \op\Rev: \ELAlg \to \op\LAut$.
    \begin{equation}
        \begin{tikzcd}
            \ELAlg & \bot & \RLAut & \bot & \LAut & \bot & \op\LAut
            \arrow["\Aut", curve={height=-18pt}, from=1-1, to=1-3]
            \arrow["\Alg", curve={height=-18pt}, from=1-3, to=1-1]
            \arrow["\Incl", curve={height=-18pt}, from=1-3, to=1-5]
            \arrow["\Reach", curve={height=-18pt}, from=1-5, to=1-3]
            \arrow["\Rev", curve={height=-18pt}, from=1-5, to=1-7]
            \arrow["\op\Rev", curve={height=-18pt}, from=1-7, to=1-5]
        \end{tikzcd}
        \label{eq:lassoAutomatonAdjunctionDiagram}
    \end{equation}
\end{theorem}

We make some observations about the adjunction.
The functor $\Rev$ maps a reachable lasso automaton to an observable lasso automaton \cite[Chap.8]{Cruchten2022TopicsInOmegaAutomata}. Informally, a lasso automaton is observable if distinct states accept distinct lasso languages. A lasso automaton is minimal if it is both reachable and observable. 
It follows that for all $(W,f,P) \in \ELAlg$, the automaton $(\Rev \circ \Incl \circ \Aut)(W,f,P)$ is observable.
Hence by taking its reachable part, we obtain a minimal automaton accepting $\Lasso(W,f,P)$.

Going in the other direction, if we start with a reachable lasso automaton $A$ accepting $L$, 
then $(\Reach \circ \op\Rev)(A)$ is a minimal automaton accepting $L^\rev$, and 
$(\Alg \circ \Reach \circ \op\Rev)(A)$ is the maximal quotient of $(\newords, \lassos)$ that recognises $L$.
This comes about as follows, cf.~\cite[Sec.~9.2]{BonchiEtAl2014AlgebraCoalgebraDualityInBrzozowski}.
The categories $\ELAlg$ and $\RLAut$ do not have initial or final objects, since
morphisms preserve the language, but if we fix a lasso language $L$ and denote by $\ELAlg(L)$ and $\RLAut(L)$ the 
full subcategories of structures that recognise, resp.~accept, $L$ then we do obtain initial and final objects
in $\ELAlg(L)$ and $\RLAut(L)$. Since $\Aut$ and $\Alg$ reverse the language, they restrict to an adjunction between
$\ELAlg(L)$ and $\RLAut({L^\rev})$. 
The final object in $\RLAut({L^\rev})$ is the minimal lasso automaton for $L^\rev$, and the final object in $\ELAlg(L)$ is 
the maximal quotient of $(\newords, \lassos)$ that recognises $L$. Since $\Alg$ is a right adjoint, it preserves final objects, hence $\Alg$ maps the minimal lasso automaton for $L^\rev$ to the maximal quotient of $(\newords, \lassos)$ that recognises $L$.

This, in particular, shows that $\Alg \circ \Reach \circ \op\Rev$ differs from Cruchten's construction~\cite[Ch.~5]{Cruchten2022TopicsInOmegaAutomata}, because the latter does not map all reachable $\Omega$-automata accepting $L$ to the maximal Wilke algebra quotient for $L$. For instance, one can observe that Cruchten's construction maps the initial $\Omega$-automaton for $L = \lassos$, with states $X = \words, Y = \lassos$, to the minimal Wilke algebra quotient $(\newords, \upwords)$.


\section{Restricting the Adjunction to \texorpdfstring{$\Omega$}{Omega}-Automata and Wilke Algebras}
\label{sec:adjunctionOmegaAutomata}
In this section, we show that the adjunction from \Cref{thm:lassoAutomataAdjunctionChain} restricts to $\Omega$-automata and Wilke algebras. First, we note that we can define a notion of extended Wilke algebra by adding a recognising subset to a Wilke algebra homomorphisms $f: (\newords, \upwords) \to W$. 
The main observation is that $\Omega$-automata are a full subcategory of $\LAut$, and extended Wilke algebras can be identified with a full subcategory of $\ELAlg$. In general, restricting an adjunction to full subcategories yields another adjunction, as long as the restricted functors are well-defined on objects. This is because hom-sets in a full subcategory are inherited from the ambient category. Therefore our task is to show that restricting the functors from \Cref{thm:lassoAutomataAdjunctionChain} to $\Omega$-automata and Wilke algebras is well-defined.

We begin with specialising extended lasso semigroups to Wilke algebras.

\begin{definition}[Extended Wilke Algebra]
    An \emph{extended Wilke algebra} is an extended lasso semigroup $(W, f, P)$ such that $W$ is a Wilke algebra, i.e., $W$ satisfies the circularity and coherence axioms. We write $\EWAlg$ for the full subcategory of $\ELAlg$ of all extended Wilke algebras.
\end{definition}

Note that in the above definition $f: (\newords, \lassos) \twoheadrightarrow W$ has the free lasso semigroup as its domain, instead of the free Wilke algebra. But, given a Wilke algebra $W$, there exists a bijective correspondence between maps of type $(\newords, \upwords) \twoheadrightarrow W$ and maps of type $(\newords, \lassos) \twoheadrightarrow W$. This correspondence is given by precomposition with the map $\phi: (\newords, \lassos) \twoheadrightarrow (\newords, \upwords)$ defined by  $\fint \phi(s) = s$ and $\inft \phi(u, v) = uv^\omega$. We prefer the type $(\newords, \upwords) \twoheadrightarrow W$, as it allows us to view $\EWAlg$ as a subcategory of $\ELAlg$.

Next, we turn to the automaton categories. As we remarked, $\Omega$-automata form a full subcategory of $\LAut$, which we write as $\OAut$. However, the functor $\Rev: \LAut \to \op\LAut$ does not restrict to $\Rev: \OAut \to \op\OAut$. In order to see why the reverse of an $\Omega$-automaton is not an $\Omega$-automaton, recall that for any $\Omega$-automaton $A$, the language $\Lasso(A)$ is saturated. But we cannot expect that $\Lasso(A^\rev) = \Lasso(A)^\rev$ is also saturated. Hence we introduce a new type of lasso automata which turn out to be exactly the reverse of some $\Omega$-automaton.

\begin{definition}[$\Omega^\rev$-automata]
    A \emph{$\Omega^\rev$}-automaton (in words, reverse-$\Omega$-auto\-maton) is a lasso automaton $A = (X, Y, q, \rho, \sigma, \xi, F)$ satisfying, for all $va, vba $ $\in \newords$ and $k > 0$:
    \begin{equation*}
        \sigma(\rho(q, v),a) = \sigma(\rho(q, (va)^kv),a) \quad \text{and} \quad \sigma(\rho(q, vb),a) = \xi(\sigma(\rho(q, av),b), a).
    \end{equation*}
    We call these identities \emph{reverse-circularity} and \emph{reverse-coherence}, respectively.
\end{definition}

\begin{apxpropositionrep}
\label{prop:RevCircularityCoherence}
    Let $A$ be a lasso automaton. If $A$ is circular, then $\Rev(A)$ is reverse-circular, and if $A$ is reverse-circular, then $\Rev(A)$ is circular. Likewise for coherence and reverse coherence.
\end{apxpropositionrep}
\begin{proofsketch}
     Let $A = (X, Y, q, \rho, \sigma, \xi, F)$ and $\Rev(A) = (X^\rev, Y^\rev,q^\rev, \rho^\rev, \sigma^\rev,$ $\xi^\rev, F^\rev)$. If $A$ is circular, $va \in \newords$ and $k > 0$:
    \begin{multline*}
        \sigma^\rev(\rho^\rev(q^\rev, v),a) = \{ x \in X \mid \xi(\sigma(x, a), v^\rev) \in F \} = \\
        = \{ x \in X \mid \xi(\sigma(x, a), v^\rev(av^\rev)^k) \in F \} = \sigma^\rev(\rho^\rev(F, (va)^kv), a),
    \end{multline*}
    where the second equality uses circularity, and the first and third equalities use the identity $\xi^\rev(\sigma^\rev(\rho^\rev (q^\rev, v),a), w) = \{ x \in X \mid \xi(\sigma(\rho(x, w^\rev),a), v^\rev) \in F \}$. Hence $\Rev(A)$ is reverse-circular. The other parts of the proposition follow by similar reasoning. \qed
\end{proofsketch}
\begin{proof}
    Let $A = (X, Y, q, \rho, \sigma, \xi, F)$ and $\Rev(A) = (X^\rev, Y^\rev,q^\rev, \rho^\rev, \sigma^\rev, \xi^\rev, F^\rev)$. From the proof of \cite[Prop.~8.22]{Cruchten2022TopicsInOmegaAutomata}, we deduce two useful identities:
    \begin{align}
        &\xi^\rev(\sigma^\rev(\rho^\rev (q^\rev, v),a), w) = \{ x \in X \mid \xi(\sigma(\rho(x, w^\rev),a), v^\rev) \in F \} \label{eq:CirCohReversalLem1} \\
        &\xi^\rev(\sigma^\rev(\rho^\rev(Z, v), a),w) \in F^\rev \iff \xi(\sigma(\rho(q, w^\rev), a), v^\rev) \in Z, \label{eq:CirCohReversalLem2}
    \end{align}
    for all $Z \in X^\rev = \mathcal P(Y)$, $v, w \in \words$, $a \in \Sigma$.
    
    If $A$ is circular, $va \in \newords$ and $k > 0$:
    \begin{multline*}
        \sigma^\rev(\rho^\rev(q^\rev, v),a) = \{ x \in X \mid \xi(\sigma(x, a), v^\rev) \in F \} = \\
        = \{ x \in X \mid \xi(\sigma(x, a), v^\rev(av^\rev)^k) \in F \} = \sigma^\rev(\rho^\rev(F, (va)^kv), a),
    \end{multline*}
    where the first and third equalities use \Cref{eq:CirCohReversalLem1} and the second equality uses circularity. Hence $\Rev(A)$ is reverse-circular. If $A$ is coherent and $vba \in \newords$:
    \begin{multline*}
        \sigma^\rev(\rho^\rev(q^\rev, vb),a) = \{ x \in X \mid \xi(\sigma(x, a),bv^\rev) \in F \} = \\
        = \{ x \in X \mid \xi(\sigma(\rho(x, a), b),v^\rev a) \in F \} = \xi^\rev(\sigma^\rev(\rho^\rev(q^\rev, av),b), a),
    \end{multline*}
    where the first and third equalities use \Cref{eq:CirCohReversalLem1} and the second equality uses coherence. Hence $\Rev(A)$ is reverse-coherent. If $A$ is reverse-circular, $av \in \newords$ and $k > 0$:
    \begin{multline*}
        \xi^\rev(\sigma^\rev(Z, a),v) \in F^\rev \iff \sigma(\rho(q, v^\rev),a) \in Z \iff \\
        \sigma(\rho(q, (v^\rev a)^k v^\rev),a) \in Z
        \iff \xi^\rev(\sigma^\rev(Z, a), v(av)^k) \in F^\rev,
    \end{multline*}
    where the first and third equivalences use \Cref{eq:CirCohReversalLem2} and the second equivalence uses reverse-circularity. Hence $\Rev(A)$ is circular. Finally, if $A$ is reverse-coherent and $abv \in \newords$:
    \begin{multline*}
        \xi^\rev(\sigma^\rev(Z, a),bv) \in F^\rev 
        \iff \sigma(\rho(q, v^\rev b),a) \in Z \iff \\
        \xi(\sigma(\rho(q, av^\rev),b), a) \in Z
        \iff \xi^\rev(\sigma^\rev(\rho^\rev(Z, a),b), va) \in F^\rev,
    \end{multline*}
    where the first and third equivalences use \Cref{eq:CirCohReversalLem2} and the second equivalence uses reverse-coherence. Hence $\Rev(A)$ is coherent.
    \qed
\end{proof}

\begin{apxpropositionrep}
\label{prop:AutAlgCircularityCoherence}
    Let $A \in \LAut$ and $(W,f,P) \in \ELAlg$. If $A$ is reverse-circular, then $\Alg(A)$ satisfies the circularity axiom, and if $(W, f, P)$ satisfies the circularity axiom, then $\Aut(W,f,P)$ is reverse-circular. Likewise for reverse-coherence and the coherence axiom.
\end{apxpropositionrep}
\begin{proofsketch}
    We show that applying $\Alg$ to a reverse-coherent automaton yields a coherent algebra. The other parts of the proposition follows by similar reasoning. Let $A = (X, Y, q, \rho, \sigma, \xi, F)$ and $\Alg(A) = (W_A, f_A, P_A)$. Suppose that $A$ is reverse-coherent and let $(\alpha_i, \beta_i,\gamma_i) \in \fint W_A$, for $i \in \{1,2\}$. We have $(\alpha_1,\beta_1,\gamma_1) = \fint f(a_1\dotsc a_n)$ and $(\alpha_2,\beta_2,\gamma_2) = \fint f(bv)$, for some $a_1, \dotsc, a_n, b \in \Sigma$, $v \in \words$. Hence:
    \begin{align*}
        &(\alpha_1, \beta_1, \gamma_1) \times \big((\alpha_2, \beta_2, \gamma_2) \cdot (\alpha_1, \beta_1, \gamma_1)\big)^\omega = (\alpha_1, \beta_1, \gamma_1) \times (\alpha_1 \alpha_2, \beta_1 \alpha_2, \gamma_1 \gamma_2)^\omega \\
        &\quad = (\alpha_1, \beta_1, \gamma_1) \times \beta_1 \alpha_2(q) = \gamma_1 \beta_2 \alpha_1 (q) = \xi(\sigma(\rho(q, a_n \dotsc a_1 v^r),b), a_n \dotsc a_1) \\
        &\quad= \xi(\sigma(\rho(q, a_{n-1}\dotsc a_1 v^\rev b), a_n), a_{n-1} \dotsc a_1) = \dotsc = \sigma(\rho(q, v^\rev b a_n \dotsc a_2), a_1) \\
        &\quad= \beta_1(\alpha_2(q)) = (\alpha_1 \alpha_2, \beta_1 \alpha_2, \gamma_1 \gamma_2)^\omega = \big((\alpha_1, \beta_1, \gamma_1) \cdot (\alpha_2, \beta_2, \gamma_2)\big)^\omega,
    \end{align*}
    where we use reverse-coherence $n$-many times in the third line. \qed
\end{proofsketch}
\begin{proof}
    Let $A = (X, Y, q, \rho, \sigma, \xi, F)$ and $\Alg(A) = (W_A, f_A, P_A)$. Suppose $A$ is reverse-circular and let $(\alpha, \beta, \gamma) \in \fint W_A$, $k > 0$. We have $(\alpha, \beta, \gamma) = \fint f_A(au)$ for some $au \in \newords$. Therefore:
    \begin{multline*}
        ((\alpha, \beta, \gamma)^k)^\omega = (\alpha^k, \beta \alpha^{k-1}, \gamma^k)^\omega = \beta \alpha^{k-1}(q) = \sigma(\rho(q, (u^\rev a)^{k-1} u^\rev), a) = \\
        = \sigma(\rho(q, u^\rev), a) = \beta(q) = (\alpha, \beta,\gamma)^\omega,
    \end{multline*}
    where we use reverse-circularity in the fourth equality. Therefore $(W_A,f_A,P_A)$ satisfies the circularity axiom. Now suppose that $A$ is reverse-coherent and let $(\alpha_i, \beta_i,\gamma_i) \in \fint W_A$, for $i \in \{1,2\}$. We have $(\alpha_1,\beta_1,\gamma_1) = \fint f(a_1\dotsc a_n)$ and $(\alpha_2,\beta_2,\gamma_2) = \fint f(bv)$, for some $a_1, \dotsc, a_n, b \in \Sigma$, $v \in \words$. Hence:
    \begin{align*}
        &(\alpha_1, \beta_1, \gamma_1) \times \big((\alpha_2, \beta_2, \gamma_2) \cdot (\alpha_1, \beta_1, \gamma_1)\big)^\omega = (\alpha_1, \beta_1, \gamma_1) \times (\alpha_1 \alpha_2, \beta_1 \alpha_2, \gamma_1 \gamma_2)^\omega \\
        &\quad = (\alpha_1, \beta_1, \gamma_1) \times \beta_1 \alpha_2(q) = \gamma_1 \beta_2 \alpha_1 (q) = \xi(\sigma(\rho(q, a_n \dotsc a_1 v^r),b), a_n \dotsc a_1) \\
        &\quad= \xi(\sigma(\rho(q, a_{n-1}\dotsc a_1 v^\rev b), a_n), a_{n-1} \dotsc a_1) = \dotsc = \sigma(\rho(q, v^\rev b a_n \dotsc a_2), a_1) \\
        &\quad= \beta_1(\alpha_2(q)) = (\alpha_1 \alpha_2, \beta_1 \alpha_2, \gamma_1 \gamma_2)^\omega = \big((\alpha_1, \beta_1, \gamma_1) \cdot (\alpha_2, \beta_2, \gamma_2)\big)^\omega,
    \end{align*}
    where we use reverse-coherence $n$-many times in the equalities on the third line. Therefore $(W_A, f_A, P_A)$ satisfies the coherence axiom.
    
    Let $\Aut(W, f, P) = (X_W, Y_W, q_W, \rho_W, \sigma_W, \xi_W, F_W)$. Suppose $(W, f, P)$ satisfies the circularity axiom and let $va \in \newords$, $k > 0$. Then:
    \begin{multline*}
        \sigma_W(\rho_W(q_W, v), a) = (\fint f(av^\rev))^\omega = (\fint f(av^\rev) \cdot \fint f((av^\rev)^k)^\omega = \\
        = (\fint f(a) \cdot \fint f(v^\rev (av^\rev)^k)^\omega = \sigma_W(\rho_W(q_W, (va)^kv),a),
    \end{multline*}
    where we use the circularity axiom in the second equality. Hence $\Aut(W,f,P)$ is reverse-circular. Now suppose $(W,f,P)$ satisfies the coherence axiom and let $vba \in \newords$. Then:
    \begin{multline*}
        \sigma_W(\rho_W(q_W, vb),a) = (\fint f(abv^\rev))^\omega = (\fint f(a) \cdot \fint f(bv^\rev))^\omega = \\
        = \fint f(a) \times (\fint f(bv^\rev) \cdot \fint f(a))^\omega = \xi_W(\sigma_W(\rho_W(q_W, av), b), a),
    \end{multline*}
    where we use the coherence axiom in the third equality. Hence $\Aut(W,f,P)$ is reverse-coherent.
    \qed
\end{proof}

Now we are ready to present the adjunction from \Cref{thm:lassoAutomataAdjunctionChain}, restricted to $\Omega$-automata and Wilke algebras.

\begin{definition}
    Write $\OAut$ for the full subcategory of $\LAut$ of all $\Omega$-automata. Write $\ROAut$ for the full subcategory of $\LAut$ of all $\Omega^\rev$-automata. Finally, write $\RROAut$ for the full subcategory of $\ROAut$ of all reachable $\Omega^\rev$-automata.
\end{definition}

\begin{theorem}
\label{thm:OmegaAutAdjunction}
    The adjunction from \Cref{thm:lassoAutomataAdjunctionChain} restricts to:
\begin{equation*}
    \begin{tikzcd}
        \EWAlg & \bot & \RROAut & \bot & \ROAut & \bot & \op\OAut
        \arrow["\Aut", curve={height=-18pt}, from=1-1, to=1-3]
        \arrow["\Alg", curve={height=-18pt}, from=1-3, to=1-1]
        \arrow["\Incl", curve={height=-18pt}, from=1-3, to=1-5]
        \arrow["\Reach", curve={height=-18pt}, from=1-5, to=1-3]
        \arrow["\Rev", curve={height=-18pt}, from=1-5, to=1-7]
        \arrow["\op\Rev", curve={height=-18pt}, from=1-7, to=1-5]
    \end{tikzcd}
\end{equation*}
\end{theorem}
\begin{proof}
    It follows from \Cref{prop:AutAlgCircularityCoherence} that the restrictions $\Aut: \EWAlg \to \RROAut$ and $\Alg: \RROAut \to \EWAlg$ are well-defined. It is straightforward to see that reverse-circularity and reverse-coherence are preserved by $\Reach$, so the restrictions $\Incl: \RROAut \to \ROAut$ and $\Reach: \ROAut \to \RROAut$ are well-defined. Finally, from \Cref{prop:RevCircularityCoherence}, we have that the restrictions $\Rev: \ROAut \to \op \OAut$ and $\op \Rev: \op \OAut \to \ROAut$ are well-defined. Therefore $\Rev \circ \Incl \circ \Aut \dashv \Alg \circ \Reach \circ \op \Rev: \EWAlg \to \op \OAut$.
\end{proof}

The observations made below \Cref{thm:lassoAutomataAdjunctionChain} apply also in the setting of \Cref{thm:OmegaAutAdjunction}, including the relationships between minimal automata and maximal quotients.
In \cite{CianciaVenema2019OmegaAutomataACoalgebraicPerspective}, a decision procedure was given for checking whether a lasso automaton is an $\Omega$-automaton. \Cref{thm:lassoAutomataAdjunctionChain} and \Cref{thm:OmegaAutAdjunction} provide an alternative algebraic procedure via the following proposition.

\begin{apxpropositionrep}\label{lem:deciding-circ-coh}
 A lasso automaton $A$ is circular and coherent iff the extended lasso semigroup $(\Alg \circ \Reach \circ \Rev)(A)$ is circular and coherent. 
 Checking whether a finite lasso semigroup  $(\fint W, \inft W)$ is circular and coherent can be done in time $O(n^2)$ where $n=|\fint W|$.
\end{apxpropositionrep}

\begin{proof}
($\Rightarrow$) follows from the well-definedness of $\Alg \circ \Reach \circ \Rev$ in \Cref{thm:OmegaAutAdjunction}.
($\Leftarrow$): By contraposition, suppose $A$ is not circular or not coherent. 
Then there exists a state $x$ in $A$ such that $\Lasso(A, x)$ is not saturated (cf. \cite[Fact 17]{CianciaVenema2019OmegaAutomataACoalgebraicPerspective}). 
Let $A'$ be the automaton obtained from $A$ by changing the initial state to $x$. Then $\Lasso(A') = \Lasso(\Alg \circ \Reach \circ \Rev(A'))$ is not saturated, and hence it cannot be both circular and coherent. But $\Alg \circ \Reach \circ \Rev(A')$ and $\Alg \circ \Reach \circ \Rev(A)$ differ only in the recognising set (recall \cite[Chap.8]{Cruchten2022TopicsInOmegaAutomata} that $\Rev$ turns the initial state into a set of final states), so $\Alg \circ \Reach \circ \Rev(A)$ is also not circular or not coherent.

Now, let $(\fint W, \inft W)$ be  a finite lasso semigroup, and let $n=|\fint W|$.
Coherence amounts to checking $n^2$ equations. 
For circularity we need to check the equation $(s^k)^\omega = s^\omega$ for all $s \in \fint W$ and all $k >0$. For fixed $s \in \fint W $, the sequence $s^1, s^2, s^3, \ldots$ has at most $n$ many distinct elements. Hence it suffices to check that these at most $n$ elements are all equal to $s^\omega$. Hence circularity requires checking for each $s \in \fint W$ at most $n$ equations plus at most $n$ lookups, giving a total of at most $n(n+n)$ checks.
Assuming lasso semigroup operations can be evaluated in constant time, we obtain an overall complexity of $O(n^2)$.

If $A$ is finite, then  $(\Alg \circ \Reach \circ \Rev)(A) = (\fint W, \inft W)$ is finite.
We provide a double-exponential upper bound on $|\fint W|$ in terms of $m=\max(|X|,|Y|)$ where $X, Y$ are the state sets of $A$.
To see this, let $(\Reach\circ\Rev)(A) = (X',Y', \ldots)$. The reverse-determinise construction gives an exponential number of states in the worst case, and the resulting automaton might be reachable. Hence $|X'|, |Y'| \leq 2^m$.
Now let $(\Alg\circ\Reach\circ\Rev)(A) = (\fint W, \inft W)$.
Using that $k^k \leq 2^{k^2}$ and $2^{(2^m)^2} = 2^{2^{2m}}$, we get that $n = |\fint W| \leq 2^{3 \cdot {2^{2m}}}$.
Hence checking circularity and coherence of $A$ via $\Alg \circ \Reach \circ \Rev(A)$ 
has an upper bound in terms of $m$ of 
$O((2^{3 \cdot {2^{2m}}})^2)=O(2^{6 \cdot {2^{2m}}})$. 
\end{proof}
The size of $(\Alg \circ \Reach \circ \Rev)(A)$ is in the worst case doubly-exponential in the number of states of $A$.
However, the exponential blow-up in the reverse-determinise construction is known to often not turn up in practice \cite{Champarnaud2002SplitAJ}, so it could be interesting to evaluate the algebraic decision procedure on some real-life examples.


\section{Conclusion}
\label{sec:conclusion}

In this paper, we introduced and studied lasso semigroups as generalisations of Wilke algebras. We proved that homomorphisms into finite lasso semigroups characterise regular lasso languages by giving language-preserving transformations between lasso automata and extended lasso semigroups. We extended these transformations to dually adjoint functors between the categories of lasso automata and of lasso semigroups extended with a recognising set, and showed that this adjunction restricts to a dual adjunction between $\Omega$-automata and extended Wilke algebras. 

Since lasso semigroups characterise regular lasso languages, we believe that they are also of interest in their own right. This is motivated by the relevance of non-saturated lasso languages (which cannot be described by a Wilke algebra) in automata learning \cite{AngluinFisman2016LearningRegularOmegaLang}.
A categorical approach to learning $\omega$-regular languages~\cite{UrbatSchroeder2020AutomataLearning}
is also based on lasso languages and algebraic recognition.
Ideas from~\cite{UrbatSchroeder2020AutomataLearning} relating language acceptance via Wilke algebras with automata acceptance provided useful inspiration for our own constructions.
A different algebraic approach to lasso languages is found in \cite{Cruchten2024KleeneTheoremsForLasso} where so-called lasso algebras are introduced as counterparts of Kleene algebra for reasoning about language equivalence of lasso expressions.

Closely related to our work is a very recent and independently developed dual adjunction~\cite{Cruchten2024OnTransitionConstructionsArxiv} between lasso/$\Omega$-automata and certain \emph{bisimulation congruences}. Bisimulation congruences correspond to lasso semigroup quotients satisfying an extra bisimulation condition. Another
difference with the present work is that the adjunction in \cite{Cruchten2024OnTransitionConstructionsArxiv}, which is based on constructions from \cite{Cruchten2022TopicsInOmegaAutomata}, is language-preserving whereas the adjunction $\Aut \dashv \Alg$ is language-reversing. We leave a detailed comparison between the two approaches as future work.

The adjunctions we have established are instrumental for clarifying the relationship between coalgebraic and algebraic approaches to languages of infinite words (although we deliberately kept the coalgebraic perspective implicit). This could aid the discovery of new coalgebraic or algebraic approaches to language theory beyond infinite words. In particular, there are extensions of Wilke algebras for \emph{infinite trees} \cite{BojIdz2009AlgebraForInfForests,IdzSkrBoj2016RegularLanguagesOfThinTrees}, but no notions of lasso or $\Omega$-automata on infinite trees. We see this as a fruitful direction for future work.

%
%
%
\bibliographystyle{splncs04}
\bibliography{references}

\begin{thebibliography}{10}
\providecommand{\url}[1]{\texttt{#1}}
\providecommand{\urlprefix}{URL }
\providecommand{\doi}[1]{https://doi.org/#1}

\bibitem{AngluinFisman2016LearningRegularOmegaLang}
Angluin, D., Fisman, D.: Learning regular omega languages. Theoretical Computer Science  \textbf{650},  57--72 (2016). \doi{10.1016/j.tcs.2016.07.031}

\bibitem{Awodey2006CategoryTheory}
Awodey, S.: Category Theory. Oxford University Press, Inc., 2nd edn. (2010), \url{https://dl.acm.org/doi/10.5555/2060081}

\bibitem{Bezhanishvili2023MinimizationInLogicalForm}
Bezhanishvili, N., Bonsangue, M.M., Hansen, H.H., Kozen, D., Kupke, C., Panangaden, P., Silva, A.: Minimisation in logical form. In: Palmigiano, A., Sadrzadeh, M. (eds.) Samson Abramsky on Logic and Structure in Computer Science and Beyond, pp. 89--127. Springer (2023). \doi{10.1007/978-3-031-24117-8_3}

\bibitem{BojIdz2009AlgebraForInfForests}
Boja{\'{n}}czyk, M., Idziaszek, T.: Algebra for infinite forests with an application to the temporal logic {EF}. In: Bravetti, M., Zavattaro, G. (eds.) CONCUR 2009 - Concurrency Theory. pp. 131--145. Springer (2009). \doi{10.1007/978-3-642-04081-8_10}

\bibitem{BonchiEtAl2014AlgebraCoalgebraDualityInBrzozowski}
Bonchi, F., Bonsangue, M.M., Hansen, H.H., Panangaden, P., Rutten, J.J.M.M., Silva, A.: Algebra-coalgebra duality in {B}rzozowski's minimization algorithm. ACM Transactions in Computational Logic  \textbf{15}(1),  1--29 (2014). \doi{10.1145/2490818}

\bibitem{CalbrixNivatPodelski1994UltimatelyPeriodicWords}
Calbrix, H., Nivat, M., Podelski, A.: Ultimately periodic words of rational $\omega$-languages. In: Brookes, S., Main, M., Melton, A., Mislove, M., Schmidt, D. (eds.) Mathematical Foundations of Programming Semantics (MFPS 1993). LNCS, vol.~802, pp. 554--566. Springer (1994). \doi{10.1007/3-540-58027-1_27}

\bibitem{Champarnaud2002SplitAJ}
Champarnaud, J.M., Khorsi, A., Parantho{\"e}n, T.: Split and join for minimizing: Brzozowski's algorithm. In: Proceedings of the Prague Stringology Conference. pp. 96--104 (2002), \url{https://www.stringology.org/event/2002/p11.html}

\bibitem{CianciaVenema2012StreamAutomataAreCoalgebras}
Ciancia, V., Venema, Y.: Stream automata are coalgebras. In: Pattinson, D., Schröder, L. (eds.) Coalgebraic Methods in Computer Science (CMCS 2012). LNCS, vol.~7399, pp. 90--108. Springer (2012). \doi{10.1007/978-3-642-32784-1_6}

\bibitem{CianciaVenema2019OmegaAutomataACoalgebraicPerspective}
Ciancia, V., Venema, Y.: {$\Omega$}-automata: A coalgebraic perspective on regular $\omega$-languages. In: 8th Conference on Algebra and Coalgebra in Computer Science (CALCO 2019). Leibniz International Proceedings in Informatics (LIPIcs), vol.~139, pp. 5:1--5:18 (2019). \doi{10.4230/LIPIcs.CALCO.2019.5}

\bibitem{Cruchten2022TopicsInOmegaAutomata}
Cruchten, M.: Topics in $\Omega$-Automata: A Journey Through Lassos, Algebra, Coalgebra and Expressions. Master's thesis, University of Amsterdam (2022), \url{https://eprints.illc.uva.nl/id/eprint/2209}

\bibitem{Cruchten2024KleeneTheoremsForLasso}
Cruchten, M.: Kleene theorems for lasso languages and $\omega$-languages. In: Chen, X., Li, B. (eds.) Theory and Applications of Models of Computation. pp. 111--123. Springer Nature Singapore (2024). \doi{10.1007/978-981-97-2340-9_10}

\bibitem{Cruchten2024OnTransitionConstructionsArxiv}
Cruchten, M.: On transition constructions for automata: A categorical perspective. Tech. rep. (2024), \url{http://arxiv.org/abs/2406.19312}

\bibitem{GreenBook}
Gr\"{a}del, E., Thomas, W., Wilke, T. (eds.): Automata logics, and infinite games: a guide to current research. Springer-Verlag (2002). \doi{10.1007/3-540-36387-4}

\bibitem{IdzSkrBoj2016RegularLanguagesOfThinTrees}
Idziaszek, T., Skrzypczak, M., Bojanczyk, M.: Regular languages of thin trees. Theory Comput. Syst.  \textbf{58}(4),  614--663 (2016). \doi{10.1007/S00224-014-9595-Z}

\bibitem{MacLane1971CategoriesWorking}
MacLane, S.: Categories for the Working Mathematician, Graduate Texts in Mathematics, vol.~5. Springer-Verlag (1971). \doi{10.1007/978-1-4757-4721-8}

\bibitem{PerrinPin2004InfiniteWords}
Perrin, D., Pin, J.E.: Infinite Words: Automata, Semigroups, Logic and Games, Pure and applied mathematics, vol.~141. Elsevier (2004). \doi{10.1017/S107989860000336X}

\bibitem{Pin:MathematicalFoundationsOfAutomataTheory}
Pin, J.E.: Mathematical foundations of automata theory, \url{https://www.irif.fr/~jep/PDF/MPRI/MPRI.pdf}, version February 18, 2022

\bibitem{Planting2013AutomataToMonoids}
Planting, A.: From Automata to Monoids and Back Again. Master's thesis, Radboud University (2013)

\bibitem{Rot16:CoalgMinInitialityFinality}
Rot, J.: Coalgebraic minimization of automata by initiality and finality. In: Birkedal, L. (ed.) Mathematical Foundations of Programming Semantics, {MFPS} 2016. Electronic Notes in Theoretical Computer Science, vol.~325, pp. 253--276. Elsevier (2016). \doi{10.1016/J.ENTCS.2016.09.042}

\bibitem{Rutten:TCS2000}
Rutten, J.J.M.M.: Universal coalgebra: A theory of systems. Theoretical Computer Science  \textbf{249}(1),  3--80 (2000). \doi{10.1016/S0304-3975(00)00056-6}

\bibitem{UrbatSchroeder2020AutomataLearning}
Urbat, H., Schr\"{o}der, L.: Automata learning: An algebraic approach. In: Proceedings of the 35th Annual ACM/IEEE Symposium on Logic in Computer Science. p. 900–914. LICS '20 (2020). \doi{10.1145/3373718.3394775}

\bibitem{Wilke93AlgTheoryForRegLanguagesFinInf}
Wilke, T.: An algebraic theory for regular languages of finite and infinite words. Int. J. Algebra Comput.  \textbf{3}(4),  447--490 (1993). \doi{10.1142/S0218196793000287}

\end{thebibliography}

\end{document}